\documentclass[11pt]{article}

\usepackage[letterpaper,margin=1.00in]{geometry}
\usepackage{amsmath, amssymb, amsthm, amsfonts}
\usepackage{thmtools,thm-restate}
\usepackage{bbm}

\usepackage{ifthen}
\usepackage{tikz}
\usetikzlibrary{positioning,decorations.pathreplacing}

\usepackage{cite}
\usepackage{appendix}
\usepackage{graphicx}
\usepackage{xcolor}
\usepackage{algorithm}
\usepackage{algorithm}
\usepackage[noend]{algpseudocode}
\usepackage{epstopdf}
\usepackage{wrapfig}
\usepackage{paralist}
\usepackage{wasysym}

\usepackage{dsfont}
\newcommand{\1}{\mathds{1}}

\usepackage[textsize=tiny]{todonotes}

\usepackage{framed}
\usepackage[framemethod=tikz]{mdframed}
\usepackage[bottom]{footmisc}
\usepackage{enumitem}
\setitemize{noitemsep,topsep=3pt,parsep=3pt,partopsep=3pt}
\usepackage[font=small]{caption}
\usepackage{xspace}

\newtheorem{theorem}{Theorem}[section]
\newtheorem{lemma}[theorem]{Lemma}
\newtheorem{meta-theorem}[theorem]{Meta-Theorem}

\newtheorem{corollary}[theorem]{Corollary}

\newtheorem{definition}[theorem]{Definition}

\definecolor{darkgreen}{rgb}{0,0.5,0}
\definecolor{darkblue}{rgb}{0,0,0.6}
\usepackage{hyperref}
\hypersetup{
    unicode=false,          % non-Latin characters in Acrobat’s bookmarks
    colorlinks=true,        % false: boxed links; true: colored links
    linkcolor=darkblue,          % color of internal links (change box color with linkbordercolor)
    citecolor=darkgreen,        % color of links to bibliography
    filecolor=magenta,      % color of file links
    urlcolor=cyan           % color of external links
}
\usepackage[capitalize, nameinlink]{cleveref}
\Crefname{remark}{Remark}{Remarks}
\Crefname{observation}{Observation}{Observations}
\algnewcommand\algorithmicswitch{\textbf{switch}}
\algnewcommand\algorithmiccase{\textbf{case}}

\algdef{SE}[SWITCH]{Switch}{EndSwitch}[1]{\algorithmicswitch\ #1\ \algorithmicdo}{\algorithmicend\ \algorithmicswitch}%
\algdef{SE}[CASE]{Case}{EndCase}[1]{\algorithmiccase\ #1}{\algorithmicend\ \algorithmiccase}%
\algtext*{EndSwitch}%
\algtext*{EndCase}%
\newcommand{\eps}{\varepsilon}

\newcommand{\CONGEST}{$\mathsf{CONGEST}$\xspace}
\newcommand{\LOCAL}{$\mathsf{LOCAL}$\xspace}

\newcommand{\poly}{\operatorname{\text{{\rm poly}}}}

\newcommand{\set}[1]{\left\{#1\right\}}

\newcommand{\R}{\mathbb{R}}

\newcommand{\calL}{\mathcal{L}}
\newcommand{\calC}{\mathcal{C}}

\newcommand{\FullOrShort}{full}

\ifthenelse{\equal{\FullOrShort}{full}}{
	
  \newcommand{\fullOnly}[1]{#1}
  \newcommand{\shortOnly}[1]{}

  }{

    \newcommand{\fullOnly}[1]{}
    \newcommand{\IncludePictures}[1]{}
   
  }

\begin{document}
\date{}
\title{Deterministic Distributed Vertex Coloring: \\ Simpler, Faster, and without Network Decomposition}
\author{
Mohsen Ghaffari \\
\small ETH Zurich \\
\small ghaffari@inf.ethz.ch
\and
Fabian Kuhn \\
\small University of Freiburg\\
\small kuhn@cs.uni-freiburg.de
}
\maketitle

\begin{abstract} We present a simple deterministic distributed algorithm that computes a $(\Delta+1)$-vertex coloring in $O(\log^2 \Delta \cdot \log n)$ rounds. The algorithm can be implemented with $O(\log n)$-bit messages. The algorithm can also be extended to the more general $(degree+1)$-list coloring problem. 

Obtaining a polylogarithmic-time deterministic algorithm for $(\Delta+1)$-vertex coloring had remained a central open question in the area of distributed graph algorithms since the 1980s, until a recent network decomposition algorithm of Rozho\v{n} and Ghaffari [STOC'20]. The current state of the art is based on an improved variant of their decomposition, which leads to an $O(\log^5 n)$-round algorithm for $(\Delta+1)$-vertex coloring. 

Our coloring algorithm is completely different and considerably simpler and faster. It solves the coloring problem in a direct way, without using network decomposition, by gradually rounding a certain fractional color assignment until reaching an integral color assignments. Moreover, via the approach of Chang, Li, and Pettie [STOC'18], this improved deterministic algorithm also leads to an improvement in the complexity of randomized algorithms for $(\Delta+1)$-coloring, now reaching the bound of $O(\log^3\log n)$ rounds.

As a further application, we also provide faster deterministic distributed algorithms for the following variants of the vertex coloring problem. In graphs of arboricity $a$, we show that a $(2+\eps)a$-vertex coloring can be computed in $O(\log^3 a\cdot\log n)$ rounds. We also show that for $\Delta\geq 3$, a $\Delta$-coloring of a $\Delta$-colorable graph $G$ can be computed in $O(\log^2 \Delta\cdot\log^2 n)$ rounds.
\end{abstract}
\setcounter{page}{0}
\thispagestyle{empty}

{   \newpage
    % \footnotesize
    \smallskip
    \hypersetup{linkcolor=blue}
    \tableofcontents
    \setcounter{page}{0}
    \thispagestyle{empty}
}
\newpage

\section{Introduction}
Graph coloring has been one of the central problems in the area of \textit{distributed graph algorithms} for over three decades. See, e.g., the Distributed Graph Coloring book of Barenboim and Elkin~\cite{barenboimelkin_book}. In this paper, we present a surprisingly simple deterministic distributed algorithm that improves on the state of the art considerably, and also leads to a faster randomized distributed algorithm. We first review the model and the state of the art, and then state our contribution.

\smallskip
\subsection{Background on the Coloring Problem}
\paragraph{Model.}
We work with the standard synchronous message passing model of distributed algorithms. The network is abstracted as an $n$-node undirected graph $G=(V, E)$, where each node represents one processor and has a unique $O(\log n)$-bit identifier. Initially, nodes do not know the topology of the network graph $G$, except for each knowing its own neighbors. Besides this, the nodes might know some global parameters (or parameter estimates), such as  the number of nodes $n$ and the maximum degree $\Delta$ (or suitable upper bounds on them). 
The communication between the processors/nodes happens in synchronous rounds, where per round each node can send one message to each of its neighbors in $G$. The variant where the messages are allowed to be of unbounded size is known as the \LOCAL model~\cite{linial1987LOCAL,peleg00} while the variant with bounded-size messages, usually $O(\log n)$ bits, is referred to as the \CONGEST model~\cite{peleg00}. At the end, each processor should know its own part of the output, e.g., in the coloring problem, the color assigned to its node. The main complexity measure for the algorithms is the number of rounds until all nodes know their output.

\smallskip
 \paragraph{State of the Art.} Here, we give a brief overview. We defer a more exhaustive review of the related work to \Cref{subsec:related}. 
 
 The celebrated work of Luby~\cite{luby86} and Alon, Babai, and Itai\cite{alon86} from the 1980s provide randomized distributed algorithms for $(\Delta+1)$-coloring with round complexity $O(\log n)$---in fact, even for the harder problem of maximal independent set. However, obtaining a deterministic algorithm with a similar round complexity has remained elusive. In his seminal work\cite{linial1987LOCAL,linial92}, and after discussing \cite{luby86, alon86}, Linial wrote 

\medskip
\begin{centering}
\lq\lq\textit{It is therefore particularly interesting to find out the best time complexity in terms
of $n$ for finding a ($\Delta$ + 1)-coloring, and in particular whether polylogarithmic time
suffices.}\rq\rq
\end{centering}
\medskip

The sentence, and especially the latter part, refers to deterministic algorithms. Shortly after, deterministic algorithms with round complexity $2^{O(\sqrt{\log n \cdot \log\log n})}$ and then $2^{O(\sqrt{\log n})}$ were presented by Awerbuch et al.\cite{awerbuch89} and Panconesi and Srinivasan\cite{panconesi-srinivasan}, respectively. These algorithms are based on a generic tool known as network decomposition, where the graph is decomposed into clusters of small diameter, which are colored with few colors and thus can then be processed efficiently in the \LOCAL model.
However, the question of obtaining a $\poly(\log n)$ time algorithm remained open for a long time. Indeed, Open Problem 11.3 in the influential 2013 book of Barenboim and Elkin on \textit{distributed graph coloring} asked for an even weaker objective: 

\medskip
\begin{centering}
 \lq\lq\emph{\textbf{Open Problem 11.3} Devise a $\Delta \cdot \poly(\log(\Delta))$-coloring in deterministic polylogarithmic time.}\rq\rq
\end{centering}
\medskip

The first resolution of this question was presented recently: Rozho\v{n} and Ghaffari~\cite{rozhonghaffari20} presented an $O(\log^7 n)$ round deterministic algorithm in the \LOCAL model that solves the more general network decomposition problem and this led to an $O(\log^7 n)$ round deterministic algorithm for $(\Delta+1)$-coloring. An improved variant of this network decomposition algorithm was presented more recently by Grunau, Ghaffari, and Rozho\v{n}~\cite{GGR20}, which for $(\Delta+1)$-coloring implies a $O(\log^5 n)$ round algorithm. Both of these coloring algorithms use large messages, as they end up gathering the topology around some node. Bamberger, Kuhn, and Maus~\cite{bamberger2020coloring} resolved this issue by presenting a \CONGEST-model  coloring algorithm for low-diameter graphs. This is inspired by the derandomization approach of Luby~\cite{luby1993removing} and Censor-Hillel et al.~\cite{censor2017derandomizing} of reducing the amount of randomness to $O(\log^2 n)$ bits and then fixing the bits one by one, via pessimistic estimators and global communication. Put together with the faster network decomposition algorithm~\cite{GGR20}, this gives an $O(\log^6 n)$ round deterministic algorithm for $(\Delta+1)$-coloring. This is the state of the art deterministic algorithm for $(\Delta+1)$-coloring in the \CONGEST model.

On the side of randomized algorithms, there has also been much progress (as we shall review later in \Cref{subsec:related}). Interestingly, the state of the art in the \LOCAL model is an $O(\log^5 \log n)$-round algorithm that follows from plugging the deterministic algorithm mentioned above in the randomized framework of Chang et al.~\cite{chang2018optimal}. It is also known that any improvement on the randomized complexity requires (and would imply) an improvement on the deterministic complexity~\cite{chang2017time}.

\subsection{Our Contribution} We present a surprisingly simple algorithm that solves the $(\Delta+1)$-vertex coloring directly, without using the recent breakthrough results on network decomposition (or their ideas)~\cite{rozhonghaffari20,GGR20}. Besides being considerably simpler than the state of the art, the algorithm is also quadratically faster and runs in $O(\log^2 \Delta \log n)$ rounds, with small $O(\log n)$-bit messages. It also extends to the list coloring generalization of the problem.
\begin{theorem}[\textbf{Informal}]\label{thm:main_informal}
There is a deterministic distributed algorithm that computes a $(\Delta+1)$-coloring in any graph with at most $n$ nodes and maximum degree at most $\Delta$  in $O(\log^2 \Delta\cdot \log n)$ rounds, and using $O(\log n)$ bit messages. 

The algorithm can be generalized to the $(\mathit{degree}+1)$-list coloring problem where each node $v$ should choose its color from a prescribed list of colors $L_v\subseteq\{1, ..., \mathcal{C}\}$ with size $|L_v|\geq d(v)+1$, where $d(v)$ denotes the degree of node $v$. Then, the complexity is $O(\log^2 \Delta\cdot \log n)$ rounds if we can use $O(\Delta\cdot \log \mathcal{C})$-bit messages, or $O(\log^2 \mathcal{C}\cdot \log n)$ rounds using $O(\log \mathcal{C})$-bit messages. 
\end{theorem}

Besides its simplicity and faster round complexity, our result is qualitatively different than prior work. For instance, using the ingredients of \Cref{thm:main_informal}, we can color a $1-\epsilon$ fraction of nodes for an arbitrarily small constant $\eps>0$ in $O(\log^2 \Delta+\log^* n)$ rounds of the \LOCAL model. Notably, this is nearly-independent of the network size $n$, and has only logarithmic dependencies on the maximum degree $\Delta$. In contrast, with the previous algorithm that follows from \cite{rozhonghaffari20, GGR20}, coloring even just an $\eps$-fraction of the nodes requires $\Theta(\log^4 n)$ rounds. A similar difference extends to the results in the \CONGEST model. 

Moreover, our approach provides a new useful structural understanding of the coloring problem, which we hope might find applications in other computational settings. The previous method was based on network decomposition~\cite{rozhonghaffari20, GGR20} and essentially relied on breaking the graph into low-diameter parts and then solving the coloring in each part by gathering the entire topology. In that sense, those methods, despite providing an efficient solution of coloring in the distributed model, did not provide any new structural understanding of the coloring problem. As we will outline in \Cref{subsec:method}, our approach is completely different and it allows us to cast the $(\Delta+1)$-coloring problem as a few instances of a simple and clean \textit{rounding problem}, which gradually turn fractional color assignments into integral color assignments, while approximately maintaining a simple quality measure. We believe that this structure is of independent interest and we are hopeful that it will find applications in other computational settings.

%The first resolution of this question was presented recently: Rozho\v{n} and Ghaffari~\cite{rozhonghaffari20} presented an an $O(\log^7 n)$ round deterministic algorithm in the \LOCAL model that solves the more general network decomposition problem and this led to an $O(\log^7 n)$ round deterministic algorithm for $(\Delta+1)$-coloring. An improved variant of this network decomposition algorithm was presented more recently by Grunau, Ghaffari, and Rozho\v{n}~\cite{GGR20}, which for $(\Delta+1)$-coloring implies a $O(\log^5 n)$ round algorithm. Both of these coloring algorithms use large messages, as they end up gathering the topology around some node. Bamberger, Kuhn, and Maus~\cite{bamberger2020coloring} resolved this issue by presenting a \CONGEST-model  coloring algorithm for low-diameter graphs. This is inspired by the derandomization approach of Luby~\cite{luby1993removing} and Censor-Hillel et al.~\cite{censor2017derandomizing} of reducing the amount of randomness to $O(\log^2 n)$ bits and then fixing the bits one by one, via pessimistic estimators and global communication. Put together with the faster network decomposition algorithm~\cite{GGR20}, this gives an $O(\log^6 n)$ round deterministic algorithm for $(\Delta+1)$-coloring. This is the state of the art deterministic algorithm for $(\Delta+1)$-coloring in the \CONGEST model.

\paragraph{Other implications, randomized coloring.} By plugging our deterministic list-coloring algorithm into the randomized coloring algorithm of Chang et al.\cite{chang2018optimal}, we can also improve the randomized complexity from the $O(\log^5\log n)$ bound of \cite{GGR20} to $O(\log^3\log n)$:
\begin{corollary}\label{crl:randomizedColoring}
There is a randomized algorithm in the \LOCAL model that computes a $(\Delta+1)$-coloring in any graph with at most $n$ nodes and maximum degree at most $\Delta$  in $O(\log^3\log n)$ rounds, with high probability\footnote{As standard, we use the phrase \textit{with high probability} to indicate that an event happens with probability at least $1-1/n^{c}$ for a desirably large constant $c\geq 2$.}.
\end{corollary}
We note that in a very recent paper, Halld\'{o}rsson, Nolin, and Tonoyan~\cite{HalldorssonNT21} give an improved randomized \CONGEST algorithm for $(\Delta+1)$-coloring. Their algorithm can use the deterministic \CONGEST algorithm that we give in \Cref{thm:main_informal} and by applying our result, they show that $(\Delta+1)$-coloring can also be solved in $O(\log^3\log n)$ rounds in the randomized \CONGEST model.

\paragraph{Other implications, coloring low-arboricity graphs and $\Delta$-coloring.} We also provide a variant of \Cref{thm:main_informal} with node weights ---e.g., coloring a subset of nodes with a constant fraction of the weights in $O(\log^2 \Delta)$ rounds of the \LOCAL model and $O(\log^2 \mathcal{C})$ rounds of the \CONGEST model--- and show that this leads to improvements for coloring graphs of low arboricity and for $\Delta$-coloring. In particular, in graphs of arboricity $a$, we show that a $(2+\eps)a$-vertex coloring can be computed deterministically in $O(\log^3\Delta\cdot\log n)$ rounds. We also show that for $\Delta\geq 3$, a $\Delta$-coloring of a $\Delta$-colorable graph $G$ can be computed deterministically in $O(\log^2\Delta\cdot\log^2 n)$ rounds. See \Cref{crl:ArbColoring} and \Cref{crl:DeltaColoring} for precise statements.

\subsection{Our Method in a Nutshell}\label{subsec:method}

We start by giving a high-level description of our algorithm for the \LOCAL model. The algorithm is based on the most basic randomized coloring method where each node repeatedly tries to get colored by choosing a uniformly random available color. More precisely, consider an edge $\set{u,v}$ between two nodes $u$ and $v$ of degree $d(u)$ and $d(v)$ and assume that $u$ and $v$ have color lists $L_u$ of size $|L_u|>d(u)$ and $L_v$ of size $|L_v|>d(v)$, respectively. If $u$ chooses a uniformly random color from $L_u$ and $v$ picks a uniformly random color from $L_v$, the probability that both nodes pick the same color is less than $1/\max\set{d(u)+1,d(v)+1}$. If every node picks a uniformly random color from its list, the expected number of monochromatic edges is less than $n$. If we are given a coloring with $O(n)$ monochromatic edges, it is simple to find a subset $S$ of nodes of size $\Theta(n)$ such that the induced subgraph $G[S]$ is properly colored.\footnote{A constant fraction of the nodes have at most O(1) monochromatic edges and of those nodes, creating a conflict graph with $\Theta(n)$ nodes and bounded degree. By picking an MIS of this conflict graph in $O(\log^* n)$ rounds~\cite{linial1987LOCAL,BEK15}, we obtain the required set $S$ of nodes that can keep their color.} The main goal of our paper is to develop an efficient deterministic distributed algorithm to assign a color $x_u\in L_u$ to each node $u$ such that the number of monochromatic edges is $O(n)$ and we can thus color a constant fraction of the nodes. Repeating $O(\log n)$ times then solves the problem of coloring all nodes of the graph.

Our general approach is based on the idea of rounding fractional assignments to integral assignments. We note that rounding ideas have been used successfully in the past for the maximal matching problem~\cite{fischer2020improved}. In the case of the coloring problem, we start with a fractional assignment of colors to the nodes and we gradually round this fractional solution to obtain an assignment of a single color to each node. For a node $u$, a fractional color assignment can just be thought of as a probability distribution over the colors in $u$'s list $L_u$. We define the cost of a fractional assignment as the expected number of monochromatic edges in the setting where each node $u$ independently picks a color according to this distribution. We say that a fractional color assignment to a node $u$ is $1/Q$-fractional if each color $c\in L_u$ is assigned to $u$ with value $a/Q$ for some integer $a\in \set{0,\dots,Q}$. Initially, each node $u$ computes an initial $1/\Theta(d(u))$-integral color assignment in which each color of $L_u$ obtains approximately equal values. As observed above, the total cost of such an assignment is $O(n)$. Now, assume that we are given a $1/(2Q)$-integral fractional color assignment. We want to turn this color assignment into an $1/Q$-integral fractional color assignment such that the expected number of monochromatic edges (essentially) does not increase. If the fractional assignments of nodes are rounded one node at a time, each node can simply decide which half of the $1/(2Q)$-fractional values to round up by $1/(2Q)$ and which half of those values to round down by $1/(2Q)$ so that the expected cost of its edges does not increase. This can be parallelized in the following way. The rounding decisions of nodes that are non-adjacent do not depend on each other. If we are given a coloring of $G$ with $C$ colors, nodes of the same color class can thus be rounded in parallel and obtain an $1/Q$-integral assignment in $O(C)$ rounds. In order for this to be useful, we however need to speed up this process significantly.

This is where we use the fact that we only do gradual rounding steps. The rounding from $1/(2Q)$-integral values to $1/Q$-integral values guarantees that each fractional value increases at most by a factor $2$. Therefore, the probability of an edge becoming monochromatic can increase at most by a factor $4$, even if the rounding is done in a worst-case way. We define edge weights $w(e)$, where $w(e)$ is the probability of the edge becoming monochromatic. By using a defective coloring algorithm of \cite{Kuhn2009WeakColoring,KawarabayashiS18}, for any parameter $\eps>0$, we can efficiently compute an $O(1/\eps^2)$-coloring of $G$ such that the total weight of monochromatic edges is at most an $\eps$-fraction of the total weight of all the edges. In \Cref{sec:defective_algorithms}, we further show that because we only need to bound the overall weight of monochromatic edges, we can even reduce the number of colors to $O(1/\eps)$. Let $E'$ be the set of monochromatic edges of this defective coloring. Because the edges in $E'$ only have an $\eps$-fraction of the total weight of all edges, even if nodes round their fractional values in a worst-case way, the total cost of the edges in $E'$ is still at most a $4\eps$-fraction of the overall cost of the original fractional assignment. On the remaining graph $G'=(V,E\setminus E')$, we have a proper $O(1/\eps)$-coloring and we can therefore round the fractional values in time $O(1/\eps)$ without loss in the quality of the fractional solution. Overall, we can get from a $1/(2Q)$-integral fractional assignment to a $1/Q$-fractional integral assignment in time $O(1/\eps)$ and at the cost of losing a $1+O(\eps)$-factor in the overall cost. By choosing $\eps=1/\log\Delta$, we can do $O(\log\Delta)$ rounding steps to get from a $1/O(\Delta)$-integral fractional assignment to an integral assignment such that the total number of monochromatic edges only increases by a factor of $(1+O(\eps))^{O(\log\Delta)}=O(1)$. The total running time is $O(\log^2\Delta)$, repeating $O(\log n)$ times gives the claimed $O(\log^2\Delta\cdot\log n)$-time bound for $(\mathit{degree}+1)$-list coloring in the \LOCAL model.

The above algorithm requires the \LOCAL model because for the rounding, nodes have to learn the complete fractional color assignment of their neighbors. For the \CONGEST model, we build on ideas that were recently developed in \cite{kuhn2020faster,bamberger2020coloring}. Assume that all nodes have a list consisting of colors from a range $\calC$. Consider an arbitrary (fixed) partition of the color space $\calC$ into a small number of parts $\calC_1,\dots,\calC_k$ of size approximately $|\calC_i|\approx |\calC|/k$. Instead of directly coloring the nodes, the goal for each node $v$ is to pick one of the parts $i\in\set{1,\dots,k}$ and update its list to $L_v\cap \calC_i$. Node $v$ then only remains in conflict with the neighbors that also pick the same color subspace $\calC_i$ and the goal is to find an assignment of color subspaces such that on average over all nodes, the ratio between remaining degree and list size does not grow by more than a $1+\eps$ factor for an appropriate choice of $\eps$. In \cite{bamberger2020coloring}, it is shown that a simple randomized choice of color subspace $\calC_i$ solves this problem in expectation. By using a generalization of the rounding process described above for coloring in the \LOCAL model, we can efficiently compute a good assignment of color spaces. With the right choice of parameters, this leads to a $(\mathit{degree}+1)$-coloring algorithm with a time complexity of $O(\log^2|\calC|\cdot \log n)$, using messages of only $O(\log|\calC|)$ bits.

\subsection{Other Related Work}\label{subsec:related}
Graph coloring and variants of it have been at the center of studies in distributed algorithms for over three decades and there is a vast amount of related work on this topic. Here, we provide a brief (and certainly not exhaustive) review of the most relevant work. We refer to the 2013 \textit{Distributed Graph Coloring} book of Barenboim and Elkin~\cite{barenboimelkin_book} for more information. 

\paragraph{Deterministic algorithms, focusing on the \boldmath$n$-dependency.}
We start with a review of the literature focusing on the $n$-dependency in the complexity. Linial~\cite{linial1987LOCAL} gave a deterministic $O(\Delta^2)$-coloring algorithm with round complexity $O(\log^*n)$. For the precise $(\Delta+1)$-coloring objective, Awerbuch et al.~\cite{awerbuch89} gave an $2^{O(\sqrt{\log n\log\log n})}$ round algorithm via network decomposition, which was improved shortly after to $2^{O(\sqrt{\log n})}$ rounds by Panconesi and Srinivasan\cite{panconesi-srinivasan}. In terms of polylogarithmic-time algorithms, the first significant progress was made by Barenboim and Elkin who gave an $O(\Delta^{1+o(1)})$ coloring in $\poly(\log n)$ rounds. 

Since then, progress on the $(\Delta+1)$-vertex coloring problem remained elusive, except for a recent $2^{O(\sqrt{\log \Delta})} \log n$-round algorithm of Kuhn~\cite{kuhn2020faster}. But much more progress, and especially in the polylogarithmic complexity regime, was made on the simpler problem of $(2\Delta-1)$ edge coloring (which is the special case of vertex coloring if we take the line graph) and its tighter variants: Ghaffari and Su gave a $\poly(\log n)$ round algorithm for $(2\Delta-1)(1+o(1))$ edge coloring. Fischer, Ghaffari, and Kuhn~\cite{FischerGK17} gave a $\poly(\log n)$ round algorithm for $2\Delta-1$ coloring. Ghaffari et al.~\cite{det_1+eps_edgecoloring} gave a $\poly(\log n)$ algorithm for $(1+o(1))\Delta$ edge coloring assuming $\Delta=\Omega(\log n)$ and a $3\Delta/2$ coloring for general graphs. Improving on the polylogarithmic complexity, Harris~\cite{harris2019distributed} gave an algorithm with round complexity $\tilde{O}(\log^2 \Delta) \cdot O(\log n)$ for $(2\Delta-1)$-edge coloring. All these developments remained confined to the edge-coloring problem and could not be extended to the harder vertex coloring problem.

The question of whether a $\poly(\log n)$ round deterministic algorithm for $(\Delta+1)$-coloring exists was finally resolved in 2020: Rozho\v{n} and Ghaffari~\cite{rozhonghaffari20} gave a $O(\log^7 n)$ round algorithm for network decomposition, which led to an algorithm with the same complexity for $(\Delta+1)$-coloring, in the \LOCAL model. Bamberger et al.~\cite{bamberger2020coloring} showed how to solve coloring using small messages in low-diameter graphs and, combining this with the network decomposition of \cite{rozhonghaffari20}, obtained a $\poly(\log n)$ round algorithm in the \CONGEST model for $(\Delta+1)$-coloring. Most recently, the network decomposition algorithm was improved in~\cite{GGR20} to complexity $O(\log^5 n)$. For the $(\Delta+1)$ coloring problem, this implied an $O(\log^5 n)$ round algorithm in the \LOCAL model and a fairly involved $O(\log^6 n)$ algorithm in the \CONGEST model.  

\paragraph{Deterministic algorithms, focusing on the \boldmath$\Delta$-dependency.} We now review the literature with a focus on the $\Delta$-dependency. Linial's $O(\log^* n)$ round $O(\Delta^2)$-coloring algorithm directly translates to an $O(\Delta^2 + \log^* n)$ round algorithm for $(\Delta+1)$-coloring. This complexity was improved to $O(\Delta \log \Delta + \log^* n)$ by Kuhn and Wattenhofer~\cite{Kuhn2006On}, and subsequently to $O(\Delta + \log^* n)$ by Kuhn~\cite{Kuhn2009WeakColoring} and independently by Barenboim and Elkin~\cite{barenboim10}. Recently, Barenboim, Elkin, and Goldenberg~\cite{BEG18} presented an alternative (locally-iterative) algorithm with the same $O(\Delta + \log^* n)$ complexity. The complexity was improved to $\tilde{O}(\Delta^{3/4}) + O(\log^* n)$ by Barenboim~\cite{barenboim2016deterministic}, and then to $\tilde{O}(\Delta^{1/2}) + O(\log^* n)$ by Fraigniaud, Heinrich, and Kosowski~\cite{fraigniaud16}. The bound has remained around this complexity, with the exception of the following improvements: the aforementioned work of Barenboim et al.\cite{BEG18} improved the bound to $O(\sqrt{\Delta\log\Delta} \log^* \Delta + \log^* n)$ and Maus and Tonoyan~\cite{MausTonoyan20} sharpened this to $O(\sqrt{\Delta \log \Delta} + \log^* n)$ rounds. These algorithms with $\tilde{O}(\Delta^{1/2}) + O(\log^* n)$ complexity use large messages (roughly, $O(\Delta)$ bits).   
For the $(2\Delta-1)$ edge coloring problem, which is special case of $(\Delta+1)$-vertex coloring, we now know much better bounds: Kuhn~\cite{kuhn2020faster} gave an $2^{\sqrt{\log \Delta}} + O(\log^* n)$ round algorithm, and the bound was recently improved to $\log^{O(\log\log \Delta)}\Delta + O(\log^* n)$ rounds by Balliu, Kuhn, and Olivetti~\cite{edgecoloring_quasipolylog}.

\paragraph{Randomized algorithms, focusing on the \boldmath$n$-dependency.}
As mentioned at the beginning, the celebrated algorithms of Luby~\cite{luby86} and Alon, Babai, and Itai~\cite{alon86} for maximal independent set gave an $O(\log n)$ round algorithm for $(\Delta+1)$-coloring, and using $O(\log n)$ bit messages. Barenboim, Elkin, Pettie, and Schneider~\cite{barenboim2016locality} improved the complexity to $O(\log \Delta) + 2^{O(\sqrt{\log \log n})}$. The complexity was further improved by Harris, Su, and Schneider~\cite{harris2016coloring} to $O(\sqrt{\log \Delta}) + 2^{O(\sqrt{\log \log n})}$, and then by Chang, Li, and and Pettie~\cite{chang2018optimal} to $O(\log^* \Delta) + 2^{O(\sqrt{\log\log n})}$. In all these results, the latter term comes from the complexity of solving coloring (in fact, the harder variant of $degree+1$ list coloring) in graphs with $O(\log n)$ nodes. And indeed, it is known that any faster randomized algorithm would imply a faster deterministic algorithm~\cite{chang2017time}. That term was improved with the network decomposition result of Rozho\v{n} and Ghaffari~\cite{rozhonghaffari20} to $\poly(\log\log n)$. The best known upper bound on the complexity for randomized $(\Delta+1)$-coloring in the \LOCAL model, prior to our work, is $O(\log^5\log n)$ rounds and follows from the improved network decomposition of \cite{GGR20} combined with the randomized coloring algorithm of Chang et al.~\cite{chang2018optimal}. The same round complexity has recently also been achieved in \cite{HKMT21} for the \CONGEST model. An even more recent paper by Halld\'{o}rsson, Nolin, and Tonoyan~\cite{HalldorssonNT21} gives an even faster randomized \CONGEST algorithm for $(\Delta+1)$-coloring. Their algorithm uses the deterministic \CONGEST algorithm that we give in this paper to improve the round complexity of the problem to $O(\log^3\log n)$, matching the best in the \LOCAL model that we give in \Cref{crl:randomizedColoring}.

\subsection{Mathematical Preliminaries and Outline}

For a graph $G=(V,E)$ and a node $v\in V$, we use $N(v)$ to denote the set of neighbors of $v$ and we use $d(v)=|N(v)|$ to denote the degree of $v$. For an integer $k\geq 1$, we define $[k]:=\set{1,\dots,k}$. Further, unless specified otherwise, when writing $\log x$, we always mean $\log_2 x$.

The remainder of the paper is organized as follows. In \Cref{sec:defective_results}, we give formal definitions of the type of weighted defective colorings we use in our algorithms and we state the formal defective coloring results we use. In \Cref{sec:labelings}, we define and analyze a problem, where each node needs to choose a label from a given domain and there is a cost associated with each edge, depending on the edge and the labels assigned to both nodes of the edge. We in particular show that a fractional assignment of labels can be efficiently rounded to an integral labeling at only a small increase in cost. Finally in \Cref{sec:algorithms}, we show how this rounding procedure can be used to obtain our coloring algorithms and prove \Cref{thm:main_informal}.

\section{Weighted Defective Coloring}
\label{sec:defective_results}

In this section, we define the notion of weighted defective coloring, review the known algorithms for it, and then state our faster algorithm for a relaxation of this problem, which we call weighted average defective coloring. This is a basic tool that we will use throughout the algorithm presented in the next section. We note that the main novelty of our approach lies in the next sections. The following weighted defective coloring variant has in particular been studied in \cite{KawarabayashiS18}.

\begin{definition}[Weighted Defective Coloring]\label{def:defcoloring}
  Given a weighted graph $G=(V,E,w)$ with non-negative edge weights $w(e)\geq 0$ for all $e\in E$, a parameter $\eps>0$, and an integer $C\geq 1$, a weighted $\eps$-relative defective $C$-coloring of $G$ is an assignment $\varphi:V\to [C]$ of colors in $[C]$ to the vertices of $G$ such that
  \[
  \forall u\in V\,:\,\sum_{v\in N(u)}\1_{\set{\varphi(u)=\varphi(v)}}\cdot
  w(u,v) \leq \eps\cdot \sum_{v\in N(u)}w(u,v).
  \]
\end{definition}

\paragraph{Known Algorithms for Weighted Defective Coloring.} Kuhn\cite{Kuhn2009WeakColoring} showed that in unweighted graphs, one can compute an $\eps$-relative coloring with $O(1/\eps^2)$ colors in $O(\log^* q)$ rounds of the \CONGEST model, if one is given an initial proper $q$-coloring of the graph. Kawarabayashi and Schwartzman~\cite{KawarabayashiS18} observed that essentially the same algorithm and analysis can be also be generalized to the weighted case (for a formal statement, see \Cref{lemma:defcoloring} in \Cref{sec:defective_algorithms}). 

\paragraph{Relaxing to Weighted Average Defective Coloring and Faster Algorithms.} We use weighted defective coloring as a basic subroutine in our algorithms. In fact, we can use even a weaker version of weighted defective coloring, where the defect is only bounded on average. For this weaker notion, we present a faster algorithm, i.e., one with a better trade-off between the weighted defect and the number of colors.

\begin{definition}[Weighted Average Defective Coloring]\label{def:avgdefcoloring}
  Given a weighted graph $G=(V,E,w)$ with non-negative edge weights $w(e)\geq 0$ for all $e\in E$, a parameter $\eps>0$, and an integer $C\geq 1$, a weighted average $\eps$-relative defective $C$-coloring of $G$ is  an assignment $\varphi:V\to [C]$ of colors in $[C]$ to the nodes of $G$ such that
  \[
  \sum_{u\in V}\,\,\sum_{v\in N(u)}\1_{\set{\varphi(u)=\varphi(v)}}\cdot
  w(u,v) \leq \eps\cdot\sum_{u\in V}\sum_{v\in N(u)} w(u,v).
  \]
\end{definition}

The above definition is equivalent to requiring that at most an $\eps$-fraction of the total weight of all the edges of $G$ is on monochromatic edges. 

The following lemma shows that for every integer $C\geq1$, there is an efficient distributed algorithm, with round complexity essentially linear in $C$, that computes a $C$-coloring such that the average weighted defect per node $v$ is almost equal to $1/C$ times the total weight of the edges of $v$. The proof of \Cref{lemma:recursiveavgdefective} can be obtained by adapting relatively standard techniques and it therefore appears at the end of the paper in \Cref{sec:defective_algorithms}.

\begin{restatable}[Weighted Average Defective Coloring]{lemma}{LemmaRecAvgDefective}
\label{lemma:recursiveavgdefective}
 Let $G=(V,E,w)$ be a weighted graph with non-negative edge weights and assume that we are given a proper vertex coloring of $G$ with colors in $[q]$. Then, for every integer $C\geq 1$ and every $\delta>0$, there is a deterministic $O(C/\delta + \log^*q)$-round algorithm to compute a weighted average $(1+\delta)/C$-relative defective $C$-coloring of $G$. The algorithm requires messages of size $O(\log q)$ bits.
\end{restatable}

We note that the time bound in \Cref{lemma:recursiveavgdefective} has an additive $\log^* q$ term. We could just use $\log^*q = \log^* n$ by using the unique IDs as colors. However, since we have to apply the weighted average defective coloring lemma several times, this would give a slightly worse overall time complexity for our coloring algorithm. In the application, we will therefore always first compute an $O(\Delta^2)$-coloring in time $O(\log^* n)$ at the beginning of the algorithm and we can afterwards have $\log^*q = \log^*\Delta$. In all our applications of the lemma, $C/\delta$ will be $\omega(\log^*\Delta)$, so that the term disappears asymptotically.

\paragraph{Remark.} We note that for an unweighted graph $G=(V,E)$, the above lemma already gives a simple way to obtain an $O(\Delta\log n)$ coloring in $\poly(\log n)$ time. In particular, 
for $C=2$, the theorem implies that the nodes $V$ can be partitioned into two parts $V_a$ and $V_b$ such that the two induced subgraphs $G[V_a]$ and $G[V_b]$ together have at most $(1+\delta)\cdot|E|/2$ edges. When iterating the same procedure on the individual parts $k=\log\Delta$ times and by setting $\delta=1/\log\Delta$, we can divide the nodes $V$ into $2^k=\Delta$ parts $V_1,\dots,V_{\Delta}$ such that all the induced subgraphs $G[V_1],\dots,G[V_{\Delta}]$ together have at most 
\[
\frac{(1+\delta)^k}{2^k} \cdot |E| = \frac{(1-1/\log\Delta)^{\log\Delta}}{\Delta} \cdot |E| < \frac{e\cdot|E|}{\Delta} = O(n)
\]
edges. Therefore a constant fraction of the nodes have at most a constant number of neighbors in their own part and by computing a maximal independent set of the bounded-degree parts of each graph $G[V_i]$, in $O(\log^*q)$ additional time, we can therefore color a constant fraction of all nodes of $G$ with $\Delta$ colors. By repeating $O(\log n)$ times, this gives an extremely simple $O(\log\Delta /\delta \cdot \log n)=O(\log^2\Delta\cdot \log n)$-round algorithm to color $G$ with $O(\Delta\log n)$ colors. Even for such a coloring, the only previously known $\poly\log(n)$-time deterministic algorithm was based on first computing a network decomposition by using the recent algorithms of \cite{rozhonghaffari20} or \cite{GGR20}.

\section{Vertex Labelings with Edge Costs}
\label{sec:labelings}

As mentioned in the overview of our technique, a key part of our coloring algorithms is the rounding step. In this section, we present our rounding approach, formalized in the context of a vertex labeling problem that we define next.

\subsection{Problem Statement}

\begin{definition}[Edge Cost Vertex Labeling]\label{def:edgecost_labeling}
  Given a graph $G=(V,E)$, in an instance of the \emph{edge cost vertex labeling} problem, we are given a finite set of labels $\calL$ and an edge cost function $c:(V\times \calL)^2 \to \R_{\geq0}$. A solution is an assignment of a label $\ell_v\in \calL$ to each node $v\in V$, where the cost $C$ of the labeling is given as
  \[  
  C := \sum_{\set{u,v}\in E} c\big((u,\ell_u),(v,\ell_v)\big).
  \]  
\end{definition}

As the main result of this section, we give an efficient deterministic algorithm that achieves the following. Given an arbitrary random assignment of labels, we show how to deterministically compute a labeling with a total cost that is almost as good as the expected cost of the random labeling. It is convenient to think of a random label assignment as a fractional labeling of the nodes, as defined below.

\begin{definition}[Fractional Vertex Labeling]\label{def:fractionallabeling}
  Given a graph $G=(V,E)$ and a label set $\calL$, a \emph{fractional vertex labeling} of $G$ is an assignment of fractional values $x_{v,\ell}\in [0,1]$ for each label $\ell$ to each node $v\in V$ such that for all $v\in V$, $\sum_{\ell\in L}x_{v,\ell}=1$. A fractional vertex labeling is called $(1/Q)$-integral for some positive integer $Q$ if all fractional values $x_{v,\ell}$ are integer multiples of $1/Q$.
\end{definition}

The edge cost of a fractional vertex labeling is given by the expected cost if each node independently picks its label according to the probability distribution given by its fractional values. 

\begin{definition}[Edge Cost of Fractional Labeling]\label{def:fractionalcost}
  Given a graph $G=(V,E)$ and an instance of the vertex labeling with edge costs problem with label set $\calL$ and edge cost function $c:(V\times \calL)^2 \to \R_{\geq0}$, the cost $C$ of a fractional labeling $x_{v,\ell}$ for $v\in V$ and $\ell\in \calL$ is defined as
  \[
  C := \sum_{\set{u,v}\in E}\sum_{(\ell_u,\ell_v)\in\calL^2} 
       x_{u,\ell_u}\cdot x_{v,\ell_v}\cdot c\big((u,\ell_u), (v,\ell_v)\big). 
  \]  
\end{definition}

\subsection{The Algorithm for Rounding Vertex Labelings}
We next show how a given fractional vertex labeling can be efficiently turned into an integer vertex labeling, while only losing a small factor in the overall edge cost. We first show how the method of conditional expectations can be used together with a given vertex coloring of a graph to round a given fractional solution without loss. We then see how to obtain a fast rounding method with only a small loss on the overall edge cost.

\begin{lemma}[Basic Rounding Lemma]\label{lemma:basicrounding}
  Let $G=(V,E)$ be a graph and assume that for some integer $Q\geq 1$ and a label set $\calL$, we are given a $1/(2Q)$-integral fractional vertex labeling $x_{v,\ell}$ and an edge cost function $c:(V\times \calL)^2 \to \R_{\geq0}$. If we are also given a proper $\gamma$-vertex coloring of $G$, there is an $\gamma$-round distributed algorithm to compute an $1/Q$-integral fractional vertex labeling $x_{v,\ell}'$ such that for all $v$ and $\ell$, $x_{v,\ell}'\leq 2\cdot x_{v,\ell}$ and such that the total cost of the fractional vertex labeling $x_{v,\ell}'$ is not larger than the cost of the fractional vertex labeling $x_{v,\ell}$. The distributed algorithm requires messages of size $O\big(\min\set{Q\cdot \log|\calL|, |\calL|\cdot\log Q}\big)$ bits.
\end{lemma}
\begin{proof}
   For a node $v\in V$, we define $\calL_v\subseteq \calL$ to be the set of labels $\ell$ for which $x_{v,\ell}$ is an integer multiple of $1/Q$ and we define $\overline{\calL}_v:=\calL\setminus \calL_v$ to be the remaining set of labels. To get a $1/Q$-integral solution, for each node $v\in V$, we do the following. If $\ell\in \calL_v$, we can just define $x_{v,\ell}':=x_{v,\ell}$ because $x_{v,\ell}$ is already $1/Q$-integral. For the labels in $\overline{\calL}_v$, first observe that for each label $\ell\in\overline{\calL}_v$, we have $x_{v,\ell}=i/(2Q)$ for some odd integer $i\geq 1$ (otherwise, $x_{v,\ell}$ would be $1/Q$-integral). Hence, we know that the set $\overline{\calL}_v$ is of even size. If $\ell\in \overline{\calL}_v$, we set $x_{v,\ell}'$ to either $x_{v,\ell}-1/(2Q)$ or to $x_{v,\ell}+1/(2Q)$. In order to make sure that $\sum_{\ell\in \calL}x_{v,\ell}'=1$ as required by \Cref{def:fractionallabeling}, for exactly half the labels $\ell\in \overline{\calL}_v$, we need to set $x_{v,\ell}'=x_{v,\ell}-1/(2Q)$ and for exactly half the labels $\ell\in \overline{\calL}_v$, we need to set $x_{v,\ell}'=x_{v,\ell}+1/(2Q)$. In this way, we clearly guarantee that $x_{v,\ell}'\leq 2\cdot x_{v,\ell}$. In the following, we describe how we choose which labels in $\overline{\calL}_v$ to round down and which labels in $\overline{\calL}_v$ to round up. 
   
   Recall that we are given a proper coloring of the nodes of $G$ with colors from $[\gamma]$. For every $i\in[\gamma]$, let $V_i$ be the set of nodes that are colored with color $i$. The algorithm consists of $\gamma$ phases $1,\dots,\gamma$, where in phase $i$, all nodes in $V_i$ choose their new fractional label assignment. In each phase $i$, we show that we can round the fractional assignments of all $v\in V_i$ such that the total edge cost of the fractional labeling does not increase. For convenience, for all nodes $v\in V$, we define a variable $y_{v,\ell}$, which we initialize to $y_{v,\ell}=x_{v,\ell}$ at the beginning of the algorithm. When processing node $v$, we decide about the new fractional value $x_{v,\ell}'$ and we set $y_{v,\ell}=x_{v,\ell}'$. Like this, the values $y_{v,\ell}$ always define the current partially rounded fractional assignment.
   
   Note that because we are given a proper $\gamma$-coloring, for every $i\in[\gamma]$, the nodes in $V_i$ form an independent set of $G$. Let us now focus on one phase $i$ and consider some node $v\in V_i$. Because $V_i$ is an independent set, for all the edges that are incident to node $v$, the change of the cost of those edges in phase $i$ only depends on how node $v$ chooses its new rounded fractional label assignment. We therefore have to show that node $v$ can pick the values $x_{v,\ell}'$ for all $\ell\in \calL$ such that the total cost of its edges does not increase. For every node $v\in V_i$, we now define some weight $W_{v,\ell}$ for every $\ell\in \overline{\calL_v}$ as follows:
   \[
     W_{v,\ell} := \sum_{u\in N(v)} \sum_{\ell'\in\calL} y_{u,\ell'}\cdot c((v,\ell), (u,\ell')).
   \]
   Note that $W_{v,\ell}$ is the cost of $v$'s edges if $v$ would fix its label to $\ell$ (and the fractional assignment of $v$'s neighbors remains fixed). If the fractional value $x_{v,\ell}$ is increased or decreased by $1/(2Q)$, the total cost of $v$'s edges therefore increases or decreases by $W_{v,\ell}/(2Q)$. In order to not increase the total weight, $v$ therefore wants to decrease the fractional value for labels $\ell$ with small $W_{v,\ell}$ and it wants to increase the fractional value for labels $\ell$ with large $W_{v,\ell}$.
   Let $L:=|\overline{\calL}_v|$ be the number of labels of node $v$ for which $x_{v,\ell}$ has to be rounded up or down. As mentioned, we know that $L$ is an even number. We define a set $\overline{\calL}_{v,-}\subseteq \overline{\calL}_v$ of size $\overline{\calL}_{v,-}=L/2$ by adding the $L/2$ labels $\ell\in \overline{\calL}_v$ with the largest weights $W_{v,\ell}$ to the set $\overline{\cal_L}_{v,-}$ (ties broken arbitrarily). This in particular implies that
   \[
    \sum_{\ell\in \overline{\calL}_{v,-}} W_{v,\ell} \geq \sum_{\ell\in \overline{\calL}_{v,+}} W_{v,\ell}
   \]
   Node $v\in V_i$ computes its new fractional assignment $x_{v,\ell}'$ as follows. For all $\ell\in \overline{\calL}_{v,-}$, we set $x_{v,\ell}':=x_{v,\ell}-1/(2Q)$ and for all $\ell\in \overline{\calL}_{v,+}$, we set $x_{v,\ell}':=x_{v,\ell}+1/(2Q)$.
   Let $C_v$ be the total cost of all edges of $v$ at the beginning of phase $i$ and let $C_v'$ be the total cost of all edges of $v$ at the end of phase $i$. We have
   \begin{eqnarray*}
     C_v' - C_v& = & \sum_{w\in N(v)}\sum_{(\ell_v,\ell_w)\in \calL^2} 
     (x_{v,\ell_v}' - x_{v,\ell_v})\cdot y_{w,\ell_w} \cdot c\big((v,\ell_v),(w,\ell_w)\big)\\
     & = & \sum_{\ell_v\in \overline{\calL}_v}(x_{v,\ell_v}'-x_{v,\ell_v})\cdot
     \underbrace{\sum_{w\in N(v)}\sum_{\ell_w\in\calL}y_{w,\ell_w}\cdot c\big((v,\ell_v),(w,\ell_w)\big)}_{=W_{v,\ell_v}}\\
     & = & \frac{1}{2Q}\cdot\left[
     \sum_{\ell_v\in \overline{\calL}_{v,+}} W_{v,\ell_v}
     - \sum_{\ell_v\in \overline{\calL}_{v,-}} W_{v,\ell_v}
     \right]\ \leq\ 0.
   \end{eqnarray*}
    The rounding therefore does not increase the total edge cost. Clearly, every phase can be implemented in a single round, each node just needs to learn the current fractional assignment of all neighbors. A fractional assignment with $1/(2Q)$-integral fractional values can be encoded with $O\big(\min\set{Q\cdot \log|\calL|, |\calL|\cdot\log Q}\big)$ bits. This prove the bound on the message size, which concludes the proof.
\end{proof}

The above lemma has a round complexity that is linear in the number of colors. Of course, if aiming for a proper coloring, this would require up to $\Delta+1$ colors. In the next lemma, we see how to get a faster algorithm, by relaxing the rounding objective to an approximate rounding, and using (weighted average) defective coloring.

\begin{lemma}[Approximate Rounding Lemma]\label{lemma:approxrounding}
  Let $\eps>0$ be a parameter and let $G=(V,E)$ be a graph and assume that for some integer $Q\geq 1$ and a label set $\calL$, we are given a $1/(2Q)$-integral fractional vertex labeling $x_{v,\ell}$ and an edge cost function $c:(V\times \calL)^2 \to \R_{\geq0}$. If we are also given a proper $\gamma$-vertex coloring of $G$, there is an $O(1/\eps + \log^*\gamma)$-round distributed algorithm to compute an $1/Q$-integral fractional vertex labeling $x_{v,\ell}'$ such that the total cost of the fractional vertex labeling $x_{v,\ell}'$ is by at most a $1+\eps$ factor larger than the cost of the fractional vertex labeling $x_{v,\ell}$. The distributed algorithm requires messages of size $O\big(\min\set{Q\cdot \log|\calL|, |\calL|\cdot\log Q}+ \log\gamma \big)$ bits.
\end{lemma}
\begin{proof}
   We define edge weights $w(e)\geq 0$, where the weight of an edge $\set{u,v}$ is the cost of the edge according to the initial fractional vertex labeling $x_{v,\ell}$:
   \[
   w\big(\set{u,v}\big) := \sum_{(\ell_u,\ell_v)\in\calL^2} x_{u,\ell_u}\cdot x_{v,\ell_v}\cdot c\big((u,\ell_u),(v,\ell_v)\big).
   \]
   For a node $v\in V$ to know the weights of its edges, it is sufficient if $v$ learns the initial fractional labelings of all its neighbors $u\in N(v)$.
   Each node $u$ can transmit this information to each of its neighbors using at most $O\big(\min\set{Q\cdot \log|\calL|, |\calL|\cdot\log Q}\big)$ bits. In the following, let $W:=\sum_{e\in E} w(e)$ be the total weight of all the edges and thus the total edge cost of the fractional labeling. 
   
   Next, we use \Cref{lemma:recursiveavgdefective} to compute a weighted $(\eps/3)$-relative average defective $O(1/\eps)$-coloring of graph $G$ with edge weights $w(e)$. By \Cref{lemma:recursiveavgdefective}, this requires $O(1/\eps + \log^* \gamma)$ rounds and messages of size $O(\log \gamma)$. Let $E_0$ be the set of monochromatic edges of this $O(1/\eps)$-coloring and let $W_0:=\sum_{e\in E_0}w(e)$ be the total weight (and thus the total cost) of the edges in $E_0$. Note that by the definition of a weighted $(\eps/3)$-relative average defective coloring, we have $W_0\leq \eps/3\cdot W$. Similarly, define $E_1\setminus E_0$ to be the bichromatic edges and let $W_1=W-W_0$ the total weight (and thus total cost) of the bichromatic edges.
   
   We now define the graph $G_1:=(V,E_1)$ as the subgraph, where we remove all monochromatic edges. Note that the defective $O(1/\eps)$-coloring that we computed is a proper coloring of $G_1$. By \Cref{lemma:basicrounding}, we can therefore compute a $1/Q$-integral fractional labeling $x_{v,\ell}'$ such that for all $v$ and $\ell$, $x_{v,\ell}'\leq 2x_{v,\ell}$ and such that the total edge cost of the new fractional labeling on graph $G_1$ (i.e., on the edges $E_1$) is at most $W_1$. By \Cref{lemma:basicrounding}, this can be done in $O(1/\eps)$ rounds and with messages of size at most $O\big(\min\set{Q\cdot \log|\calL|, |\calL|\cdot\log Q}\big)$ bits.
   
   The total edge cost of the new fractional labeling is the total cost of the edges in $E_1$, which is upper bounded by $W_1$, and the total cost of the edges in $E_0$. For an edge $\set{u,v}\in E$, let $w'(\set{u,v})$ be the cost of the edge according to the new fractional labeling, i.e.,
   \[
   w'\big(\set{u,v}\big) := \sum_{(\ell_u,\ell_v)\in\calL^2} x_{u,\ell_u}'\cdot x_{v,\ell_v}'\cdot c\big((u,\ell_u),(v,\ell_v)\big).
   \]
   Because we have $x_{v,\ell}'\leq 2\cdot x_{v,\ell}$ for all $v$ and $\ell$, we have $w'(e)\leq 4w(e)$ for all $e\in E$. The new total cost $W_0'$ of the edges in $E_0$ after the rounding is therefore at most $4 W_0$. Let $W'$ be the total cost of all edges after rounding. We have $W'\leq W_1 + 4 W_0 \leq W - W_0 + 4W_0 = W + 3W_0$. We already know that $W_0\leq \eps/3\cdot W$ and we thus have $W'\leq (1+\eps)W$, which concludes the proof.
\end{proof}

\begin{corollary}\label[corollary]{crl:labeling}
    Let $\eps>0$ be a parameter and let $G=(V,E)$ be a graph and assume that for some integer $k\geq 1$, we are given a $1/2^k$-integral fractional vertex labeling $x_{v,\ell}$ and an edge cost function $c:(V\times \calL)^2 \to \R_{\geq0}$. If a proper $\gamma$-coloring of $G$ is given, there is an $O(k^2/\eps+\log^* \gamma)$-round algorithm to compute an integral vertex labeling with a total edge cost at most $(1+\eps)$ times the total edge cost of the given fractional labeling. The algorithm uses messages of $O\big(\min\set{2^k\cdot \log|\calL|, |\calL|\cdot k} + \log \gamma\big)$ bits.
\end{corollary}
\begin{proof}
    For convenience, we define $Q:=2^k$. We define $\delta:= \eps /(2Q^2)$ and we first compute a weighted $\delta$-defective $O(1/\delta^2)$-coloring in time $O(\log^* \gamma)$ by using \Cref{lemma:defcoloring}, where the edge weights are defined by the edge costs of the initial fractional labeling. Note that even if all nodes first need to learn the edge costs of their edges, the algorithm can be carried out with messages of size $O\big(\min\set{Q\cdot \log|\calL|, |\calL|\cdot\log Q} + \log \gamma\big)$ bits. Note also that even if fractional values are rounded to integer values in an arbitrary way, because we start with an $1/Q$-integral solution, the cost of an edge can grow by at most a factor $Q^2$. Even after rounding, the total cost of all the monochromatic edges of the computed defective coloring is $Q^2\cdot \delta = \eps/2$ times the total initial cost of all the edges. We can therefore remove the monochromatic edges of the defective coloring and we now have to round fractional values in a properly $O(Q^4/\eps^2)=O(16^k/\eps^2)$-colored graph. In this graph, we can still afford to lose a factor $1+\eps/2$. We can now do this rounding by applying \Cref{lemma:approxrounding} $k$ times to get from a $1/2^k$-integral fractional assignment to an integral fractional assignment. If we use parameter $\eps' = c\cdot \eps / k$ for a sufficiently small constant $c$ when applying \Cref{lemma:approxrounding}, the total cost grows by a factor at most $(1+\eps')^k \leq 1+\eps/2$ over all $k$ rounding steps. By \Cref{lemma:approxrounding}, the running time of each step is $O(1/\eps' + \log^* (16^k/\eps^2))=O(k/\eps)$ and the required message size is $O\big(\min\set{2^k\cdot \log|\calL|, |\calL|\cdot k}+ \log \gamma\big)$. Because we have $k$ rounding steps, the overall time complexity is $O(k^2/\eps)$.
\end{proof}

\section{Coloring Algorithms}
\label{sec:algorithms}

In this section, we explain how we use the vertex-labeling procedure explained in the previous section to obtain efficient vertex color algorithms.

\subsection{\textsf{LOCAL} Model Algorithm}
\label{sec:LOCALalgorithm}

Consider a $(\mathit{degree}+1)$-list coloring problem: each node $v$ is given a list of colors $L_v \subseteq \{1, \dots, \mathcal{C}\}$, with the guarantee that $|L_v|\geq d(v)+1$ where $d(v)$ denotes the degree of node $v$. The objective is to assign each node $v$ a color $\ell\in L_v$ such that any two neighboring nodes have different colors. 

\begin{theorem}\label{thm:LOCAL-coloring} There is a deterministic distributed algorithm that solves any $(\mathit{degree}+1)$-list coloring problem in $O(\log^2 \Delta\cdot \log n)$ rounds, using messages of size $O(\Delta \cdot\log(\mathcal{C}) )$. Here, $\mathcal{C}$ denotes the size of the entire space of colors, i.e., each node $v$ has a color list $L_v \subseteq \{1, \dots, \mathcal{C}\}$.
\end{theorem}
\begin{proof}                                                                                                                                                                                      
We describe an algorithm that colors a constant fraction of nodes, in $O(\log^2 \Delta)$ rounds. The theorem then follows by repeating this procedure $O(\log n)$ times, each time removing the colors of the colored nodes from the lists of their neighbors. 

As a helper tool, at the very beginning, we compute a $O(\Delta^2)$ proper coloring $\phi$ of all vertices. This can be done in $O(\log^* n)$ time using Linial's classic algorithm\cite{linial1987LOCAL}. \footnote{\label{note1} As we will later see, this $O(\Delta^2)$ proper coloring means that computing a new coloring of a subgraph of these vertices can be done with a complexity that is $O(\log^* \Delta)$, which will be subsumed by higher complexities of the algorithm such as $O(\log^2 \Delta)$. This way, we avoid paying a $O(\log n\log^* n)$ term in the complexity.}

We now describe how we color a constant fraction of nodes. We start with a simple fractional color assignment and then use the rounding procedure described in the previous section to turn this into an integral color assignment with essentially the same cost. We then see how, thanks to the fact that the overall cost is small, we can properly color a constant fraction of the nodes. 

\smallskip
\paragraph{Fractional Color Assignment.} Consider a node $v$ with color list $L_v$ and degree $d(v)$. Let $\hat{d}(v)$ be the largest power of $2$ that is less than or equal to $d(v)+1$, i.e., $\hat{d}(v) = 2^{\lfloor \log_2 d(v)+1 \rfloor}$. Notice that $\hat{d}(v) \leq d(v)$ and ${d(v)}/2\leq {\hat{d}(v)}$. 
Let $\hat{L}_v$ be an arbitrary subset of $L_v$ with size $\hat{d}(v)$. Let $\calL = \{1, \dots, \mathcal{C}\}$. Consider the fractional assignment $x_v$ where for each color $\ell \in \hat{L}_v$, we set $x_{v, \ell} = \frac{1}{\hat{d}(v)}$ and for every color $\ell \notin \hat{L}_v$, we set  $x_{v, \ell} = 0$. Notice that this is a $1/Q$-integral assignment for a value of $Q=2^{\lfloor \log_2 (\Delta +1) \rfloor}$ and $Q=O(\Delta)$. Moreover, for each edge $e=\{u,v\}$, define the cost function 
\begin{equation*}
c\big((u, \ell_u),(v, \ell_v)\big) = \begin{cases}
1 &\text{if $\ell_u= \ell_v$}\\
0 &\text{if $\ell_u\neq \ell_v$}
\end{cases}
\end{equation*}
 Observe that for these fractional assignment and edge costs, we can upper bound the total cost as
\begin{align}\nonumber
C &= \sum_{\set{u,v}\in E}\sum_{(\ell_u,\ell_v)\in\calL^2} 
       x_{u,\ell_u}\cdot x_{v,\ell_v}\cdot c\big((u,\ell_u), (v,\ell_v)\big)  \\\nonumber
    &= \sum_{\set{u,v}\in E} \sum_{\ell \in \hat{L}_u\cap \hat{L}_v} x_{u,\ell}\cdot x_{v,\ell}
    \qquad\quad = \sum_{\set{u,v}\in E} \sum_{\ell \in \hat{L}_u\cap \hat{L}_v} \frac{1}{\hat{d}(u) \cdot \hat{d}(v)}    \\\nonumber
    &= \sum_{\set{u,v}\in E} |\hat{L}_u\cap \hat{L}_v| \cdot\frac{1}{\hat{d}(u) \cdot \hat{d}(v)} 
    \ \ \ = \frac{1}{2}\cdot\sum_{u\in V} \sum_{v \in N(u)} |\hat{L}_u\cap \hat{L}_v| \cdot\frac{1}{\hat{d}(u) \cdot \hat{d}(v)} \\\nonumber
    & \leq \frac{1}{2}\cdot\sum_{u\in V} \sum_{v \in N(u)} \hat{d}(v)\cdot \frac{1}{\hat{d}(u) \cdot \hat{d}(v)} 
    \ \leq  
    \frac{1}{2}\cdot\sum_{u\in V} \sum_{v \in N(u)} \frac{1}{\hat{d}(u)}  \\\label{eq:LOCALinitialcost}
    &= \frac{1}{2}\cdot\sum_{u\in V} \frac{d(u)}{\hat{d}(u)}\ \ \leq\ \ \frac{1}{2}\cdot\sum_{u\in V} 2\ \ =\ \ n.
    \end{align}     
As a side comment, one can view the above fractional assignment as a randomized algorithm where each node $v$ picks a uniformly random color in $\hat{L}_v$ and the cost function is simply the expected number of monochromatic edges --- those edges that both endpoints get the same color --- under this random color assignment.  

\smallskip
\paragraph{Integral Color Assignment.} By applying \Cref{crl:labeling}, with $\eps=1$ and $Q=O(\Delta)$, we get an integral assignment with cost at most $2n$, in $O(\log^2 \Delta)$ rounds and using messages of size $O(\Delta \log(\mathcal{C}) )$ bits. Once we have integral assignments, for each node $v$, we have $x_{v, \ell} =1$ for exactly one color $\ell \in \hat{L}_v$ and we have $x_{v, \ell'}=0$ for any other color $\ell'$. Hence, each node $v$ is assigned exactly one color $\ell\in \hat{L}_v \subseteq L_v$. Moreover, in this case, the cost $$\sum_{\set{u,v}\in E}\sum_{(\ell_u,\ell_v)\in\calL^2} x_{u,\ell_u}\cdot x_{v,\ell_v}\cdot c\big((u,\ell_u), (v,\ell_v)\big)$$ is simply the total number of monochromatic edges, i.e., the total number of edges $e=\{u, v\}$ whose both endpoints $u$ and $v$ received the same color $\ell$ as their integral color assignment. Hence, we have a color allocation to the nodes such that the number of monochromatic edges is at most $2n$. But this coloring might be improper and neighboring nodes might have the same color. 
       
\smallskip    
\paragraph{Proper Coloring.} We now show that we can compute a proper coloring for at least a $1/10$ fraction of the nodes. Let $S$ be the set of nodes that have at most $4$ monochromatic edges incident on them. We can conclude that $|S|\geq n/2$. Let $E'$ be the set of monochromatic edges with both points in $S$, and let $\mathcal{H}'=(S, E')$. We compute a maximal independent set $S'\subseteq S$ in $\mathcal{H}'$, in $O(\log^* \Delta)$ rounds. This can be done easily by using Linial's algorithm\cite{linial1987LOCAL} for this constant degree graph $\mathcal{H}'$. The runtime is $O(\log^* \Delta)$ rounds, instead of $O(\log^* n)$ rounds, because we employ the $O(\Delta^2)$ coloring $\phi$ computed at the very beginning as the initial set of colors in Linial's algorithm\cite{linial1987LOCAL}. This is exactly the point that we discussed in footnote \ref{note1}.
       
Notice that since $\mathcal{H}'$ has maximum degree at most $4$, we have $S'\geq |S|/5 \geq n/10$. We consider the integral assignments of nodes of $S'$ as their permanent color. Observe that two nodes $u, v\in S'$ might be neighboring in $G$ but then their integral color assignment must be on different colors, as otherwise the edge between them would be included in $E'$ and would be in contradiction with $S'$ being an independent set of  $\mathcal{H}'=(S, E')$. Hence, in $O(\log^2 \Delta + \log^* \Delta) = O(\log^2 \Delta)$ rounds, we managed to color at least a $1/10$ fraction of the nodes with colors from their lists such that no two neighboring nodes received the same color. 

\smallskip       
\paragraph{Wrap Up.} The above finishes the description of the procedure for coloring a constant fraction of the nodes. As mentioned at the beginning, the overall algorithm works by $O(\log n)$ repetitions of this procedure, each time on the nodes that remain uncolored. In each iteration, whenever we find permanent colors for some nodes, we remove the color permanently taken by each node from the list of the neighbors that remain uncolored. Since each permanently colored node takes away one color and one unit from the remaining degree, for each node $v$ that remains uncolored, the condition $|L_v| \geq d(v)+1$ remains satisfied, where $L_v$ is now updated to be the set of remaining colors and $d(v)$ is the number of remaining neighbors.
\end{proof}

\Cref{thm:LOCAL-coloring} directly implies a $(\mathit{degree}+1)$-list coloring algorithm in $O(\log^2\Delta\cdot\log n)$ rounds of the \LOCAL model. Plugging this into the randomized algorithm of Chang et al.~\cite{chang2018optimal} improved the randomized complexity of computing $(\Delta+1)$-coloring.

\bigskip
\noindent \emph{\textbf{\Cref{crl:randomizedColoring} (restated).} There is a randomized algorithm in the \LOCAL-mode that computes a $(\Delta+1)$-coloring in any graph with at most $n$ nodes and maximum degree at most $\Delta$  in $O(\log^3\log n)$ rounds, with high probability.}
\medskip

\begin{proof}[Proof Sketch] The randomized part of the $(\Delta+1)$-coloring algorithm of Chang et al.~\cite{chang2018optimal} runs in $O(\log^* n)$ rounds and colors most of the nodes, with the following guarantee: with high probability, each of the connected components in the subgraph induced by the nodes that remain uncolored has size $O(\log n)$. Then, we run the deterministic algorithm \Cref{thm:LOCAL-coloring} on each of these components separately. Notice that for each remaining node $v$, its list of remaining colors $L_v$ --- i.e., those colors in $\{1, \dots, \Delta+1\}$ that are not already used by colored neighbors of $v$---  has size $|L_v|\geq d(v)+1$, where $d(v)$ denotes the number of neighbors of $v$ that remain uncolored. This is because, at the start $v$ had at most $\Delta$ neighbors and among the $\Delta+1$ colors of $\{1, \dots, \Delta+1\}$, each already colored neighbor of $v$ used one color. Hence, the coloring problem in each connected component is an instance of $(\mathit{degree}+1)$-list coloring. Since the component size is $N=O(\log n)$, the algorithm of \Cref{thm:LOCAL-coloring} runs in $O(\log^3 N) = O(\log^3\log n)$ rounds. Overall, this is a round complexity of $O(\log^* n + \log^3 \log n) =O(\log^3 \log n)$.
\end{proof}

\subsection{\textsf{CONGEST} Model Algorithm}
\label{sec:CONGESTalgorithm}

\begin{theorem}\label{thm:CONGEST-coloring} There is a deterministic distributed algorithm that solves any $(\mathit{degree}+1)$-list coloring problem in $O(\log^2 \mathcal{C}\cdot \log n)$ rounds, using messages of size $O(\log \mathcal{C})$ bits. Here, $\mathcal{C}$ denotes the size of the entire space of colors, i.e., each node $v$ has a color list $L_v \subseteq \{1, \dots, \mathcal{C}\}$. 
\end{theorem}

\begin{proof}
Similar to the proof of \Cref{thm:LOCAL-coloring}, we describe a procedure that colors a constant fraction of nodes in $O(\log^2 \mathcal{C})$ rounds. The theorem then follows by repeating this procedure $O(\log n)$ times, each time removing the colors of the colored nodes from the lists of their neighbors. 

And again, as a helper tool, at the very beginning, we compute a $O(\Delta^2)$ proper coloring $\phi$ of all vertices. This can be done in $O(\log^* n)$ time using Linial's classic algorithm\cite{linial1987LOCAL}.  See again footnote \ref{note1} for why we compute this coloring. 

\paragraph{Outline of Coloring a Constant Fraction of the Nodes.} We now describe how we color a constant fraction of the nodes. 
Let $K=\lfloor\sqrt{\log \mathcal{C}}\rfloor$. We devise a partitioning algorithm for the color space, with $H=\log_K \mathcal{C}=O(\log \mathcal{C}/\log\log \mathcal{C})$ partition levels. In the first partition level, we partition the space of colors $\{1, \dots, \mathcal{C}\}$ into $K$ equally sized parts. More generally, for each $i\in\{2, \dots, H\}$, we partition each level-$(i-1)$ part of colors into $K$ equal sized parts. Hence, at the end of $H=\log_K \mathcal{C}$ levels of partitioning, each part is simply one color.

Per iteration, after we partition the space of colors, each node chooses to follow exactly one part. Let us focus on the very first iteration. In this iteration, we partition $\{1, \dots, \mathcal{C}\}$ into $K$ equal size parts. Let us call these parts $P_1, \dots, P_K$ where $\bigcup_{i=1}^k P_i = \{1, \dots, \mathcal{C}\}$. Now, each node $v$ will choose one of these $K$ parts (in a manner that will be described). If node $v$ choose part $\ell \in\{1,\dots, K\}$, then we update its color list to be $L_{v} \cap P_\ell$. Moreover, we update the graph by removing edges whose endpoints chose different parts, i.e., we keep only edges for which both endpoints are in the same part. The idea is that the node will choose one of the colors $L_{v} \cap P_\ell \subseteq P_\ell$ and thus it will not be in conflict with neighbors who are choosing their colors from other parts $P_{\ell'}$ for $\ell'\neq \ell$.

\paragraph{How does a node choose its part?} We will come back to describing the algorithm for choosing the parts. Let us for now describe the guarantee provided by this algorithm and discuss why this guarantee is useful. For each node $v$, define its cost $\Phi(v)=d(v)/|L_v|$. Hence, at the very beginning, the total cost over all nodes is $\sum_{v\in V} d(v)/|L_v| \leq \sum_{v\in V}d(v)/(d(v)+1) < n$. The guarantee of the partitioning algorithm is that, after the partition, the total cost $\sum_{v\in V} d(v)/|L_v|$ grows by at most a $1+1/H$ factor. In the new cost, now $d(v)$ is only the neighbors in the same part $\ell$ as $v$ and $L_v$ is the colors in its original list that are in its chosen part $P_{\ell}$. 

We repeat this partitioning for $H=\log_K \mathcal{C}$ levels, such that per level the total cost grows by at most a $1+1/H$ factor. At the end, each part is a single color and each node is in one of these parts. Moreover, we can upper bound the cost after all these levels $\sum_{v\in V} d(v)/|L_v| \leq n (1+1/H)^{H} < 3n$. We next discuss how this small cost allows us to properly color a constant fraction of nodes.

Let $S$ be the set of nodes $v$ for which $\Phi(v)= d(v)/|L_v| \leq 4$. Since $\sum_{v\in V} \Phi(v) < 3n$, we have $|S|\geq n/4$. At this point, we can properly color a constant fraction of $S$ in a similar way as we did in the proof of \Cref{thm:LOCAL-coloring}. In particular, let $E'$ be the set of monochromatic edges --- edges whose both endpoints have the same single-color list--- with both points in $S$, and let $\mathcal{H}'=(S, E')$. We compute a maximal independent set $S'\subseteq S$ in $\mathcal{H}'$, in $O(\log^* \Delta)$ rounds. This can be done easily by using Linial's algorithm\cite{linial1987LOCAL} for this constant degree graph $\mathcal{H}'$. The runtime is $O(\log^* \Delta)$ rounds, instead of $O(\log^* n)$ rounds, because we employ the $O(\Delta^2)$ coloring $\phi$ computed at the very beginning as the initial set of colors in Linial's algorithm\cite{linial1987LOCAL}. Notice that since $\mathcal{H}'$ has maximum degree at most $4$, we have $S'\geq |S|/5 \geq n/20$. We consider the color assignments of nodes of $S'$ as their permanent color. Hence, in $O(\log^2 \Delta + \log^* \Delta) = O(\log^2 \Delta)$ rounds, we managed to color at least a $1/20$ fraction of the nodes with colors from their lists such that no two neighboring nodes received the same color. As in \Cref{thm:LOCAL-coloring}, by $O(\log n)$ repetitions of this procedure, we can color all nodes in $O(\log^2 \Delta \log n)$ rounds. 

\paragraph{Partitioning Algorithm.} What remains from the above outline is to explain how each one level of color partitioning works, and in particular, how each node chooses its part so that the total cost function $\sum_{v\in V} d(v)/|L_v|$ does not increase by more than a $1+1/H$ factor. We describe the process for the very first level of partitioning; the same process is applies in all other levels. We first describe a fractional assignment of nodes to the parts and then explain how to turn this to an integral assignment, using the vertex labeling algorithm of the previous section.

\paragraph{Fractional Partition Assignment.}  For each node $v$ and each $\ell\in \{1, \dots, K\}$, define the fractional assignment $x_{v, \ell} = \frac{|L_v\cap P_\ell|}{|L_v|}$. That is, $x_{v, \ell}$ is the fraction of the list of node $v$ that is in the $\ell^{th}$ part. 
As a side comment, one can view this fractional assignment as the probabilities in a randomized algorithm where node $v$ chooses each part $\ell$ with probability $x_{v, \ell} = \frac{|L_v\cap P_\ell|}{|L_v|}$. Notice that $\sum_{\ell} x_{v, \ell}=1.$ Moreover, define the cost function of each edge $e=\{u,v\}$ as follows: 

\begin{equation*}
c\big((u, \ell_u),(v, \ell_v)\big) = \begin{cases}
\frac{1}{|L_v \cap P_{\ell}|} +\frac{1}{|L_u \cap P_{\ell}|} &\text{if $\ell_u= \ell_v = \ell$}\\
0 &\text{if $\ell_u\neq \ell_v$}
\end{cases}
\end{equation*}
We can now observe that for this fractional assignment and the edge costs defined above, the total cost over all the edges is exactly equal to the cost $\Phi=\sum_{v\in V} \frac{d(v)}{|L_v|}$. The argument is as follows:
\begin{align}\nonumber
& \sum_{\set{u,v}\in E}\sum_{(\ell_u,\ell_v)\in\calL^2} 
   x_{u,\ell_u}\cdot x_{v,\ell_v}\cdot c\big((u,\ell_u), (v,\ell_v)\big)   \\ \nonumber
= & \sum_{\set{u,v}\in E}\sum_{\ell=1}^{K} x_{u,\ell_u}\cdot x_{v,\ell_v}\cdot c\big((u,\ell), (v,\ell)\big) \\\nonumber
=  &\sum_{\set{u,v}\in E}\sum_{\ell=1}^{K}  \frac{|L_u\cap P_\ell|}{|L_u|} \cdot  \frac{|L_v\cap P_\ell|}{|L_v|} \cdot \bigg( \frac{1}{|L_v \cap P_{\ell}|} +\frac{1}{|L_u \cap P_{\ell}|}\bigg)  \\ \nonumber
=  &\sum_{\set{u,v}\in E}\sum_{\ell=1}^{K} \bigg(\frac{|L_u\cap P_\ell|}{|L_v| \cdot |L_u| }+ \frac{|L_v\cap P_\ell|}{|L_v| \cdot |L_u| }\bigg) \\\nonumber
= & \sum_{\set{u,v}\in E} \bigg(\frac{\sum_{\ell=1}^{K} |L_u\cap P_\ell|}{|L_v| \cdot |L_u| } + \frac{\sum_{\ell=1}^{K} |L_v\cap P_\ell|}{|L_v| \cdot |L_u| }\bigg) \\\nonumber
= & \sum_{\set{u,v}\in E} \bigg(\frac{|L_u|}{|L_v| \cdot |L_u| } + \frac{|L_v|}{|L_v| \cdot |L_u| }\bigg) \\\label{eq:initialpotential}
= & \sum_{\set{u,v}\in E} \bigg(\frac{1}{|L_v| } + \frac{1}{|L_u|}\bigg)\ =\ \sum_{v\in V} \frac{d(v)}{|L_v|}.
\end{align}

\paragraph{Turning the Fractional Assignment to an Integral Assignment.} We now would like to turn the above fractional assignment to an integral assignment.
First note that for an integral assignment of parts, if each node $v$ is assigned one of the parts $P_{\ell_v}$, the total cost is equal to
\begin{equation}\label{eq:integrallistcost}
    \sum_{\set{u,v}\in E} \1_{\ell_u=\ell_v}\cdot \left(
    \frac{1}{|L_u\cap P_u|} + \frac{1}{|L_v\cap P_v|}
    \right) =
    \sum_{u\in V}\sum_{v\in N(u)}\frac{\1_{\ell_u=\ell_v}}{|L_u\cap P_{\ell_u}|} = \sum_{u\in V}\frac{d'(u)}{|L_u'|},
\end{equation}
where $d'(u)$ is the number of neighbors of $u$ that choose the same part and $L_u'$ is the new list of $u$. The cost is therefore equal to the potential function for the new problem with reduced color space and lists.

The idea for rounding the fractional color space assignment to an integral color space assignment is to apply the algorithm of the previous section, as summarized in \Cref{crl:labeling}. There is only one issue that should be resolved before. Observe that in the fractional assignment above, the fractional values $x_{v, \ell} = \frac{|L_v\cap P_\ell|}{|L_v|}$ can be an arbitrarily small (positive) value. Thus, we cannot guarantee that these fractional values are $1/Q$-integral for a moderately small value of $Q$. We next discuss how to fix this issue.

Let $\rho$ be the greatest power of $2$ that is less than or equal to $1/(10KH)$. That is, $\rho = 2^{-\lceil\log_2 10KH\rceil}$. Notice that $1/\rho = 2^{\lceil\log_2 10KH\rceil}$ is an integer. We now slightly perturb the fractional values $x_{v, \ell}$ in a way that (1) each $x_{v, \ell}$ becomes a multiple of $\rho$, (2) we maintain $\sum_{\ell} x_{v, \ell}=1$, and (3) each $x_{v, \ell}$ is increased by at most a $(1+1/(4H))$ factor. 

First, we round down each $x_{v, \ell}$ to the largest multiple of $\rho$. That is, we define $x'_{v, \ell} =  \rho \lfloor x_{v, \ell}/\rho\rfloor$. This satisfies properties (1) and (3) but breaks property (2). To fix that, define the overall residue as $r=\sum_{\ell} \big(x_{v, \ell} - x'_{v, \ell}\big)$. 
We would like to add this residue in integer multiples of $\rho$ to $x'_{v, \ell}$. We allow each $x'_{v, \ell}$ to increase to at most $x'_{v, \ell} + \rho\lfloor \frac{x'_{v, \ell}/(4H)}{\rho}\rfloor \leq x'_{v, \ell} + x'_{v, \ell}/(4H)$. That is, we have $\rho\lfloor \frac{x'_{v, \ell}/(4H)}{\rho}$ capacity for increase for $x'_{v, \ell}$. Notice that this capacity is an integer multiple of $\rho$. Moreover, these capacities over all the labels provide sufficient room for admitting all the residue addition and getting back to satisfy property (2). The reason is as follows. For each $\ell$, we have $x_{v, \ell} - x'_{v, \ell}<\rho$. Thus, the overall residue $r=\sum_{\ell} \big(x_{v, \ell} - x'_{v, \ell}\big)$ is at most $1/(10H)$. On the other hand, our total capacity for the residues is at least $$\sum_{\ell} \big(x'_{v, \ell}/(4H) - \rho \big) = \big(\sum_{\ell} x'_{v, \ell}/(4H)\big) - K\rho \geq \frac{1-1/(10H)}{4H} - \frac{K}{10KH} \geq \frac{1}{10H} \geq r.$$ Hence, we can allocate all of the residue to the labels $\ell$ in integer multiples of $\rho$, such that each $x'_{v, \ell}$ increases to at most $x'_{v, \ell} + (\rho)\lfloor \frac{x'_{v, \ell}/(2H)}{\rho}\rfloor \leq x'_{v, \ell} (1+1/(4H))$. Let $x''_{v, \ell} $ be the assignment after these increases for label $\ell$. At this point, (1) each $x''_{v, \ell}$ is a multiple of $\rho$, (2) we have $\sum_{\ell} x''_{v, \ell}=1$, and (3) each $x''_{v, \ell} \leq x'_{v, \ell} (1+1/(4H)) \leq x_{v, \ell} (1+1/(4H))$. 

Because of the third property, we can conclude that the potential function increases by at most a $(1+1/(4H))^2$ factor during this step. Because of property (1), now each fractional assignment is $\rho$-integral, i.e., it is $1/Q$-integral for $Q=1/\rho$. Now, we invoke the algorithm of \Cref{crl:labeling} on these fractional values, setting $Q=1/\rho = O(KH)$ and $\eps = 1/(2H)$. We get an algorithm with round complexity $O(\log^2{Q}/\eps) = O(H\log^2(KH)) = O(\log \mathcal{C} \cdot \log \log \mathcal{C})$ where the last inequality follows from $H=\Theta(\log \mathcal{C} / \log \log \mathcal{C} )$ and $K= \Theta(\sqrt{\log \mathcal{C}})$. The message size is $O(\mathcal{L} \log Q) = O(K \log{KH}) = O(\sqrt{\log \mathcal{C}} \cdot \log\log \mathcal{C}) = O(\log \mathcal{C}).$ By \Cref{eq:integrallistcost}, the cost after rounding is equal to the new potential and we therefore now have an integral assignment of one part among $\{1,\dots, K\}$ in a way that the potential $\sum_{v\in V}d(v)/|L_v|$ increased by at most a factor of $(1+1/(4H))^2 (1+1/(2H)) \leq (1+1/H)$. This completes the description for one level of partitioning.

\paragraph{Wrap Up.} As mentioned before, we repeat the above process of partitioning the color space by a $K$ factor for $H=\Theta(\log \mathcal{C} / \log \log \mathcal{C} )$ repetitions. This takes $O(\log \mathcal{C} \cdot \log \log \mathcal{C})$ rounds per level and thus $O(\log^2 \mathcal{C})$ rounds overall. At the end, each part of the color space is a single color and we have assignments of nodes to parts (single colors) such that $\sum_{v\in V}d(v)/|L_v| \leq n(1+1/(H))^{H} < 3n$. At this point, as explained above, we can find a set $S$ of at least $n/4$ nodes $v$ each with potential $\Phi(v)\leq 4$ and then a maximal independent set $S'$ of size at least $n/20$ among these which then get colored permanently. 

The above completes the procedure for coloring a constant fraction of nodes. As described before, via $O(\log n)$ repetitions of this procedure, we color all the nodes. Since each repetition takes $O(\log^2 \mathcal{C})$ rounds, the overall round complexity is $O(\log^2 \mathcal{C} \log n)$ rounds. Moreover, we use messages of size at most $O(\log \mathcal{C})$. This completes the proof of \Cref{thm:CONGEST-coloring}.
\end{proof}

\begin{corollary}\label[corollary]{crl:Delta+1-coloring}
There is a deterministic distributed algorithm that on, any $n$-node graph with maximum degree at most $\Delta$, computes a $(\Delta+1)$-vertex coloring in $O(\log^2 \Delta \log n)$ rounds and using messages of size $O(\log n)$ bits.
\end{corollary}
\begin{proof}
The claim follows immediately from \Cref{thm:CONGEST-coloring} by defining the list $L_v$ of each node $v$ as $L_v=\{1, \dots, \Delta+1\}$, which clearly satisfies $|L_v|\geq d(v)+1$, and defining the space of colors $\{1, \dots, \mathcal{C}\}= \{1, \dots, \Delta+1\}$. 
\end{proof}

\section{Improved Algorithm for Layered Graphs and Implications}
\label{sec:layeredgraphs}

\newcommand{\dup}{\widehat{\delta}}
\newcommand{\Dup}{\widehat{\Delta}}

We next prove an improved upper bound for layered graphs of the following form. Later, in \Cref{crl:ArbColoring} and \Cref{crl:DeltaColoring}, we use this structure to provide improved algorithms for coloring low-arboricity graphs and for $\Delta$-coloring. 

\subsection{Layered Graphs}
\begin{definition}[Layered Graph] \label{def:layeredGraph}
  A graph $G=(V,E)$ is called a \emph{$h$-layered graph} for some $h\geq 1$ if the nodes are partitioned into $h$ layers $V_1,\dots,V_h$. The \emph{up-degree} $\dup(v)$ of a node $v\in V_i$ for some $i\in \set{1,\dots,h}$ is defined as $\dup(v):=|N(v)\cap (V_i\cup\dots\cup V_h)|$, i.e., the number of neighbors in $v$'s own and in higher layers.
\end{definition}

Note that if we are given a layered graph $G=(V_1\cup\dots\cup V_h,E)$ and an assignment of color lists $L_v$ such that for for $v\in V$, $|L_v|\geq\dup(v)+1$, the corresponding list coloring problem can be solved by a simple greedy algorithm. We go through $h$ phases $p=0,\dots,h-1$, where in phase $p$, we color the nodes in $V_{h-p}$. In phase $p$, every node $v\in V_{h-p}$ has at most $\dup(v)$ neighbors that are either already colored (i.e., in higher layers) or that are in the currently active layer $V_{h-p}$. The problem of coloring each node in $V_{h-p}$ with an available color from its list is therefore a $(\mathit{degree}+1)$-coloring problem. By using our algorithms from \Cref{sec:algorithms}, this therefore also directly implies that we can deterministically solve the given list coloring problem in $O(h\cdot\log^2\Dup\cdot\log n)$ rounds in the \LOCAL model, where $\Dup$ is the maximum up-degree of $G$. In the following, we show that this round complexity can be improved to $O(\log^2\Dup\cdot \log n + h\cdot \log^3\Dup)$. 

\subsection{Coloring a Constant Fraction of the Node Weights}
In the following, for a graph $G=(V,E)$, a given fractional vertex labeling $x_{v,\ell}$, and a corresponding edge cost function $c:(V\times \calL)^2 \to \R_{\geq0}$, we define the \emph{local edge cost} $c_u$ of a node $u\in V$ as
\[
c_u := \sum_{v\in N(u)}\sum_{(\ell_u,\ell_v)\in\calL^2}
x_{u,\ell_u}\cdot x_{v,\ell_v}\cdot c\big((u,\ell_u),(v,\ell_v)\big).
\]
That is $c_u$ is the total cost of all edges incident to node $u$. We say that a fractional vertex labeling has local edge cost at most $\mu$ for some $\mu\geq 0$ if $c_u\leq \mu$ for all $u\in V$.

\begin{lemma}\label{lemma:weightedcoloringLOCAL}
  Let $G=(V,E)$ be a properly $O(\Delta^2)$-colored graph with vertex weights $w_v\geq 0$ and maximum degree at most $\Delta$. Given a $(\mathit{degree}+1)$-list coloring instance of $G$, there is a deterministic $O(\log^2\Delta)$-round \LOCAL algorithm to compute a valid partial solution for the given list coloring instance such that the total weight of all colored nodes is at least half of the total weight of all the nodes of $G$.
\end{lemma}
\begin{proof}
  Note first that it is sufficient to show that we can compute a partial coloring for which the total weight of the colored nodes is an arbitrary constant fraction of the total weight of all the nodes. The claim of the lemma then follows by running this algorithm $O(1)$ times, each time for the $(\mathit{degree}+1)$-list coloring instance induced by the uncolored nodes of $G$.
  
  Let $x_{v,\ell}$ be a fractional color assignment so that for each node $v\in V$, $x_{v,\ell}>0$, only if color $\ell$ appears in $v$'s list $L_v$ of colors. We consider the edge cost function $c$ for which $c\big((u,\ell_u),(v,\ell_v)\big)=1$ if $\ell_u=\ell_v$ and the cost is equal to $0$ otherwise. That is, the cost of an edge is equal to $1$ is the edge is monochromatic and it is equal to $0$ otherwise.
  In \Cref{sec:LOCALalgorithm} (proof of \Cref{thm:LOCAL-coloring}), we show that every node $v$ can locally (without communication) compute a $2^{-\lfloor\log_2\Delta\rfloor}$-fractional assignment of colors from its list $L_v$ such that for every node $v\in V$, $c_v\leq 2$ (cf.\ \Cref{eq:LOCALinitialcost}). 
  
  We define a new edge cost function $c':(V\times \calL)^2 \to \R_{\geq0}$ as follows 
  \[
  \forall \set{u,v}\in E,\ \forall (\ell_u,\ell_v)\in \calL^2\,:\,
  c'\big((u,\ell_u),(v,\ell_v)\big) :=
  (w_u+w_v)\cdot c\big((u,\ell_u),(v,\ell_v)\big).
  \]
  That is, if the edge $\set{u,v}$ is monochromatic, the edge cost is $w_u+w_v$ and the edge cost is $0$ otherwise. Note that the total edge cost $C'$ w.r.t.\ the new edge cost function $c'$ can be bounded as
  \begin{eqnarray}\nonumber
    C' & = & \sum_{\set{u,v}\in E}\sum_{(\ell_u,\ell_v)\in\calL^2}
    x_{u,\ell_u}\cdot x_{v,\ell_v}\cdot (w_u+w_v)\cdot
    c\big((u,\ell_u),(v,\ell_v)\big)\\ \nonumber
    & = & \sum_{u\in V}\sum_{v\in N(u)}\sum_{(\ell_u,\ell_v)\in\calL^2}
    x_{u,\ell_u}\cdot x_{v,\ell_v}\cdot w_u\cdot
    c\big((u,\ell_u),(v,\ell_v)\big)\\ \label{eq:initialweight}
    & = & \sum_{u\in V} w_u\cdot c_u\ \leq\ 2\cdot W,
  \end{eqnarray}
  where $c_u$ denote the local edge cost of node $u$ w.r.t.\ the original edge cost function $c$ and where $W:=\sum_{v\in V} w_v$ denote the total weight of all nodes.
  
  To obtain an integral color assignment of $G$, we apply \Cref{crl:labeling} with $\eps=1/3$ and with the new edge cost function $c'$. For every node $u\in V$, assume that $\ell_u$ is the color assigned to node $u$. By \Cref{crl:labeling}, the round complexity for computing the integral color assignment is $O(\log^2\Delta)$ and we know that the computed integral labeling has at most $4/3$ times the cost of the original fractional labeling. We therefore have
  \begin{equation}\label{eq:roundedweight}
  \sum_{\set{u,v}\in E} (w_u+w_v)\cdot c\big((u,\ell_u),(v,\ell_v)\big) =
  \sum_{\set{u,v}\in E} c'\big((u,\ell_u),(v,\ell_v)\big) \leq
  \frac{4}{3}\cdot C'
  \stackrel{\eqref{eq:initialweight}}{\leq}
  \frac{8}{3}\cdot W.
  \end{equation}
  For a node $u\in V$, let $d_v$ be the number of neighbors $v$ for which $\ell_v=\ell_u$ (i.e., the number of neighbors that selected the same color). Let $V_H:=\set{u\in V : d_u\geq 4}$ be the set of nodes with at least $4$ monochromatic edges and let $W_H:=\sum_{v\in V_H} w_v$ be the total weight of the nodes in $V_H$. We have
  \begin{eqnarray*}
    \frac{8}{3}\cdot W & \stackrel{\eqref{eq:roundedweight}}{\geq} & 
    \sum_{\set{u,v}\in E} (w_u+w_v)\cdot c\big((u,\ell_u),(v,\ell_v)\big)\\
    & = & \sum_{u\in V} w_u\cdot d_u\ \geq\ \sum_{u\in V_H} w_u\cdot d_u
    \ \geq \ 4\cdot W_H.
  \end{eqnarray*}
  We therefore have $W_H\leq 2W/3$ and the total weight of node $v$ with $d_v<4$ (and thus $d_v\leq 3$) is therefore at least $W/3$. Let $V_L:=V\setminus V_H$ be the nodes $d_v<4$ and for each color $\ell$, let $V_L(\ell)$ be the nodes in $V_L$ that chose color $\ell$. For every $\ell$, the graph $G[V_L(\ell)]$ has maximum degree at most $3$. In Theorem 1 of \cite{KawarabayashiKS20}, it is shown that in vertex-weighted graphs of maximum degree $\Delta$, an $O(\Delta)$-approximation for the maximum weighted independent set problem can be computed in the time to compute an MIS in such graphs. In each graph, $G[V_L(\ell)]$, we can therefore select an independent set of weight at least a constant fraction of the total weight of all node in $V_L(\ell)$ in time $O(\log^*\Delta)$ (recall that we assumed an initial proper $O(\Delta^2)$-coloring of $G$). By selecting such a set for each color class, a constant fraction of the weight in $V_L$ and therefore also a constant fraction of the weight in $V$ can be colored.
\end{proof}

\begin{lemma}\label{lemma:weightedcoloringCONGEST}
  Let $G=(V,E)$ be a properly $O(\Delta^2)$-colored $n$-node graph with vertex weights $w_v\geq 0$. Given a $(\mathit{degree}+1)$-list coloring instance of $G$ with lists from a color space of size $\calC$, there is a deterministic $O(\log^2\calC)$-round distributed algorithm to compute a valid partial solution for the given list coloring instance such that the total weight of all colored nodes is at least half of the total weight of all the nodes of $G$. If each node initially knows the weights of all its neighbors, the algorithm uses messages of size at most $O(\log\calC + \log n)$ bits.
\end{lemma}
\begin{proof}
  Note again that as in \Cref{lemma:weightedcoloringLOCAL}, it is sufficient to show that we can compute a partial coloring for which the total weight of the colored nodes is an arbitrary constant fraction of the total weight of all the nodes. The proof is in large part analogous to the proof of \Cref{thm:CONGEST-coloring}, we therefore concentrate on the differences to the proof of \Cref{thm:CONGEST-coloring}.
  
  As in the proof of \Cref{thm:CONGEST-coloring}, we define $K=\lfloor\sqrt{\log\calC}\rfloor$ and we define a hierarchical partitioning of the color space of size $\calC$ into $H=\log_K C=O(\log\calC / \log\log\calC)$ levels. On each level, we partition the color subspace of the part into $K$ approximately equal parts so that at the end of $H=\log_K\calC$ levels, each part consists of a single color. The algorithm then consists of $H$ iterations, where in each iteration, every node $v\in V$ picks a one of the parts at the current partitioning level. We describe this for the first level in more detail. Assume that the color space $\calC$ is partitioned into $K$ parts $P_1,\dots,P_K$ such that $P_1\cup\dots\cup P_K=\calC$. Each node $v$ now needs to choose one part $\ell_v\in \set{1,\dots,K}$ and after choosing the part, the node updates its color list $L_v$ to $L_v\cap P_{\ell_v}$.
  
  We again define a potential $\Phi(v)$ for each node $v\in V$, but in the weighted case, we define the potential of $v$ as $\Phi(v)=w_v\cdot d(v)/|L_v|$. Because for all $v\in V$, $|L_v|\geq d(v)+1$, the initial total potential is $\sum_{v\in V}\Phi(v)<W$, where $W:=\sum_{v\in V} w_v$. As in \Cref{thm:CONGEST-coloring}, we define a fractional assignment of parts to the nodes by $x_{v,\ell}=\frac{|L_v\cap P_{\ell}}{|L_v|}$ and we define an edge cost function, which in the weighted case is
  \begin{equation*}
    c\big((u, \ell_u),(v, \ell_v)\big) = \begin{cases}
    \frac{w_v}{|L_v \cap P_{\ell}|} +\frac{w_u}{|L_u \cap P_{\ell}|} &\text{if $\ell_u= \ell_v = \ell$}\\
    0 &\text{if $\ell_u\neq \ell_v$}.
  \end{cases}
  \end{equation*}
  As in the proof of \Cref{thm:CONGEST-coloring}, we can observe that with the given fractional assignment and edge costs, the total cost is equal to the sum of the node potentials, i.e.,
  \[
  \sum_{\set{u,v}\in E}\sum_{(\ell_u,\ell_v)\in\calL^2} 
   x_{u,\ell_u}\cdot x_{v,\ell_v}\cdot c\big((u,\ell_u), (v,\ell_v)\big) = 
   \sum_{v\in V} \frac{w_v\cdot d(v)}{|L_v|}.
  \]
  The calculations are analogous to the calculations in \Cref{eq:initialpotential} and are therefore omitted. After rounding the fractional assignment into an integer assignment where node $v$ is assigned part $P_{\ell_v}$, the total cost is equal to
  \begin{equation*}
    \sum_{\set{u,v}\in E} \1_{\ell_u=\ell_v}\cdot \left(
    \frac{w_u}{|L_u\cap P_u|} + \frac{w_v}{|L_v\cap P_v|}
    \right) =
    \sum_{u\in V}\sum_{v\in N(u)}\frac{\1_{\ell_u=\ell_v}\cdot w_u}{|L_u\cap P_{\ell_u}|} = \sum_{u\in V}\frac{w_u\cdot d'(u)}{|L_u'|},
  \end{equation*}
  where $d'(u)$ is the number of neighbors of $u$ that choose the same part and $L_u'$ is the new list of $u$. As in the unweighted case, the cost is therefore equal to the potential function for the new problem with reduced color space and lists.
  
  The rounding from a fractional to an integral assignment now works in exactly the same was as in the proof of \Cref{thm:CONGEST-coloring} by applying \Cref{crl:labeling}. After going through all the $O(\log\calC/\log\log\calC)$ levels of the partitioning, we end up with an assignment of a single color to each node $v$ such that the total potential is at most $3$ times the initial potential and thus less than $3W$. As in \Cref{thm:CONGEST-coloring}, the total round complexity is $O(\log^2\calC)$. For a node $v\in V$, let $d_v$ be the number of neighbors that are assigned the same color. The final potential is then equal to $\sum_{v\in V} w_v\cdot d_v < 3W$. We now have the same condition as in \Cref{eq:roundedweight} (with a slightly different constant) and we can therefore conclude the proof in exactly the same way as in the proof of \Cref{lemma:weightedcoloringLOCAL}. The argument works in the \CONGEST model because the distributed weighted maximum independent algorithm of \cite{KawarabayashiKS20} works in the \CONGEST model.
\end{proof}

\subsection{Coloring Layered Graphs}

\begin{lemma}\label{lemma:LayeredColoringLOCAL}
Consider a layered graph with $h$ layers as defined in \Cref{def:layeredGraph}, and an assignment of color lists $L_v$ for each $v\in V$ such that $|L_v|\geq\dup(v)+1$. There is a deterministic distributed algorithm in the \LOCAL model that solves the the corresponding list coloring problem in $O(\log^2\Dup\cdot \log n + h\cdot \log^3\Dup)$ rounds. 
\end{lemma}
\begin{proof} 
Consider the layered graph $G=(V,E)$  with $h$-layers $V_1,\dots,V_h$. We define a weight for each node. For that purpose, let us orient each edge between two different layers $V_i$ and $V_j$, for $j\geq i+1$, toward the higher layer $V_j$. Edges whose both endpoints are in the same layer $V_i$, for some $i$, are left unoriented. Notice that each oriented path in the graph has length at most $h$. For each maximal oriented path, we add a weight of $2^{k}-1$ to the $k$-node on the oriented path. Notice that these weights can be computed easily in $O(h)$ rounds of \LOCAL model, where in round $i\in [1, h]$, each node in $V_i$ sets its weight equal to the total summation of the weights received from lower layers, and then sends its weight to each neighbor in a higher layer $V_j$ for $j\geq i+1$.

 We color the nodes in iterations. Per iteration, we call a node $v\in V_i$ for some $i\in \set{1,\dots,h}$ \textit{free} if all its neighbors in higher layers $N(v)\cap (V_{i+1}\cup\dots\cup V_{h})$ are already colored. Notice that the free nodes and their remaining lists (i.e., colors in their lists that are not used by already colored neighbors in higher layers $N(v)\cap (V_{i+1}\cup\dots\cup V_{h})$) make a $(\mathit{degree}+1)$-list coloring instance. Moreover, per iteration, free nodes hold at least a constant fraction of the total weight. The reason is that for each maximal oriented path, the last node is a free node. Moreover, since the weights added along this path are exponentially increasing--- i.e., weight a weight of $2^{k}$ to the $k$-node on the oriented path---the last node has received at least half of the weight added to all nodes for this path. Hence, by invoking \Cref{lemma:weightedcoloringLOCAL}, we get a coloring that colors a subset of the free nodes that hold at least $1/4$ of the the total weight. Hence, per iteration, the total weight reduces by a factor of $3/4$. 
 
 We can bound the total summation of the weights by $n\Dup^h \cdot 2^{h}$ because each of the $n$ nodes is the source of at most $\Dup^{h}$ distinct oriented paths, and each oriented path has a total of at most $\sum_{k=1}^{h} 2^{k}-1 \leq 2^{h}$ weight. Since per iteration the total weight reduces by a factor of $3/4$, the algorithm terminates in $O(\log^2\Dup\cdot \log n + h\cdot \log^3\Dup)$ rounds. 
\end{proof}

\begin{lemma}\label{lemma:LayeredColoringCONGEST}
Consider a layered graph as defined in \Cref{def:layeredGraph}, and an assignment of color lists $L_v$ for each $v\in V$ such that $|L_v|\geq\dup(v)+1$. There is a deterministic distributed algorithm in the \CONGEST that solves the the corresponding list coloring problem in $O(\log^2\mathcal{C}\cdot \log n + h\cdot \log^2{\mathcal{C}}\cdot\log\Dup)$ rounds, using messages of size at most $O(\log\calC + \log n)$ bits.
\end{lemma}
\begin{proof}[Proof Sketch]
The proof is the same as in \Cref{lemma:LayeredColoringLOCAL} except for two changes: (1) when coloring free nodes, we invoked the \CONGEST model algorithm of \Cref{lemma:weightedcoloringCONGEST}, instead of the \LOCAL model algorithm of \Cref{lemma:weightedcoloringLOCAL}. (2) When computing the initial weights, we need slightly more care to ensure that we can compute the weights in $h$ rounds using $O(\log n)$-bit messages. In the \CONGEST model, each node $v$ with weight $w(v)$ sends $2^{\lceil w(v)\rceil}$ to its neighbors in higher layers, using only $O(\log \log w(v))$ bits to encode the exponent of $2$. Notice that the latter is at most $O(\log\log w(v))=O(\log\log (n\Dup^h))=O(\log(h\log \Dup + \log n))=O(\log n)$ bits, as $h\leq n$ and $\Dup \leq n$. Each node still sets its weight equal to the summation of the weights it receives from lower layer nodes. We can easily see that this rounding up procedure is safe, in the sense that (1) still in each oriented path the last node receives at least half of the weight of the path and therefore, free nodes hold at least half of the total weight, (2) the total weight increases by at most a $2^{h}$ factor and thus the number of rounds remains asymptotically as before.
\end{proof}

\subsection{Implications on Coloring Low-Arboricity Graphs and \boldmath$\Delta$-Coloring}
\begin{corollary}\label{crl:ArbColoring}
There are deterministic distributed algorithm that compute a $(2+\eps)a$-list-coloring coloring of any graph with arboricity at most $a$ in $O(\log^3{a}\cdot \log(n)/\eps)$ rounds of the \LOCAL model and in  $O(\log^2\mathcal{C}\cdot \log n + \log^2{\mathcal{C}}\cdot \log a\cdot \log(n)/\eps)$ rounds of the \CONGEST model, using messages of size at most $O(\log\calC + \log n)$ bits.
\end{corollary}
\begin{proof}[Proof of \Cref{crl:ArbColoring}]
First, we compute an $H$-partition of the graph in $O(\log n/\eps)$ rounds, using the standard peeling algorithm  of Barenboim and Elkin~\cite{barenboimelkin_book}. This is a layered graph with $O(\log(n)/\eps)$ layers and with $\Dup=(2+\eps) a$. Then, we invoke \Cref{lemma:LayeredColoringCONGEST}.
\end{proof}

\begin{corollary}\label{crl:DeltaColoring}
There is a deterministic distributed algorithm in the \LOCAL model that, for any $\Delta\geq 3$ and any $\Delta$-colorable $n$-node graph with maximum degree at most $\Delta$, computes a $\Delta$-coloring in $O(\log^2 \Delta \cdot\log^2 n)$ rounds.
\end{corollary}
\begin{proof}
  Let $G=(V,E)$ be a graph with maximum degree $\Delta\geq 3$ and assume that $G$ is not a complete graph (i.e., $G$ is $\Delta$-colorable).
  We use the same basic algorithm as in \cite{panconesi95delta,GhaffariHKM18}. In \cite{panconesi95delta,GhaffariHKM18}, for $\Delta\geq 3$, it is shown that if in a $\Delta$-colorable $n$-graph $G$ all, except one node $v$ are properly colored with colors from $\set{1,\dots,\Delta}$, then one can compute a $\Delta$-coloring of all nodes of $G$ by only changing the colors in an $R=O(\log_\Delta n)$-hop neighborhood of $G$. To use this, we first compute a set $S$ of nodes such that any two nodes in $S$ are at distance at least $2R+1$ and such that every node in $G$ is within distance at most $O(R\log n)=O(\log_\Delta n \cdot\log n)$ from a node in $S$. By using a ruling set algorithm from \cite{awerbuch89,SEW13}, such a set $S$ can be computed deterministically in $O(R\log n)=O(\log_\Delta n \cdot\log n)$ rounds in the \LOCAL model. The goal then is to compute a partial $\Delta$-coloring of $G$ such that all nodes except the nodes in $S$ are colored. Once this is done, because any two uncolored nodes are at least $2R+1$ hops apart, we can then convert the partial coloring into a full $\Delta$-coloring of $G$ in $R=O(\log_\Delta n)$ rounds. It therefore remains to show that we can compute the partial coloring of all nodes except the nodes in $S$ in $O(\log^2 \Delta \cdot\log^2 n)$ rounds.
  
  For this, we define a layered graph with layers $1,\dots,O(R\log n)$ as follows. The nodes in layer $V_i$ are all nodes for which the nearest node in $S$ is at distance exactly $i$. Note that the layered graph contains all nodes in $V\setminus S$. Because every node $v$ in layer $i\geq 2$ has at least one neighbor in layer $i-1$ and because every node in layer $1$ has a neighbor in $S$, for each node $v$, we have $\dup(v)\leq \Delta-1$. If we assign the color list $L_v=\set{1,\dots,\Delta}$ to each node $v$, we can use \Cref{lemma:LayeredColoringLOCAL} to compute the partial coloring of the nodes in $V\setminus S$ in time $O(\log^2\Delta\cdot \log n + R\log n\cdot \log^3\Delta)=O(\log^2\Delta\cdot\log^2 n)$.
\end{proof}

\section{The Algorithm for Weighted Average Defective Coloring}
\label{sec:defective_algorithms}

In this section, we give a distributed algorithm to solve the weighted average defective coloring problem and we prove \Cref{lemma:recursiveavgdefective}. 

\subsection{Basic Weighted Average Defective Coloring Algorithms}

We start with the weighted version of a standard distributed defective coloring result.
Recall that for a parameter $\eps>0$ and an integer $C\geq 1$, a weighted $\eps$-relative defective $C$-coloring of a weighted graph $G=(V,E,w)$ is a $C$-coloring of the nodes $V$ such that for all nodes $v\in V$, the total weight of $v$'s monochromatic edges is at most an $\eps$-fraction of the total weight of the total weight of $v$'s edges (cf.\ \Cref{def:defcoloring}). By adapting an algorithm and analysis of  \cite{Kuhn2009WeakColoring}, Kawarabayashi and Schwartzman in \cite{KawarabayashiS18} show how to efficiently compute a weighted $\eps$-relative defective coloring with $O(1/\eps^2)$ colors.

\begin{lemma}[Weighted Defective Coloring]\label{lemma:defcoloring}\cite{KawarabayashiS18}
  Let $G=(V,E,w)$ be a weighted graph with non-negative edge weights, let $\eps\in(0,1]$ be a parameter, and assume that the nodes of $G$ are properly colored with colors in $[q]$. Then, there is a distributed $O(\log^*q)$-round algorithm to compute a weighted $\eps$-relative $O(1/\eps^2)$-coloring of $G$. If initially, every node knows the edge weights of all incident edges, the algorithm requires messages of $O(\log q)$ bits.
\end{lemma}

Note that we could directly use \Cref{lemma:defcoloring} instead of \Cref{lemma:recursiveavgdefective} in all our distributed vertex coloring algorithms. This would however lead to coloring algorithms that are slower by two log-factors. In order to get a better trade-off between the defect and the number of colors of the defective coloring, we need to resort to the weaker weighted average defective coloring. Recall that for a parameter $\eps>0$ and an integer $C\geq 1$, a weighted $\eps$-relative average defective $C$-coloring of a weighted graph $G=(V,E,w)$ is a $C$-coloring of the nodes $V$ such that the total weight of all monochromatic edges is at most an $\eps$-factor of the total weight of all the edges of $G$. We first give a simple iterative algorithm to compute a weighted average defective coloring.

\begin{lemma}\label{lemma:simpleavgdefective}
  Let $G=(V,E,w)$ be a weighted graph with non-negative edge weights and assume that we are given a proper vertex coloring with colors in $[q]$ of $G$. Then, for every integer $C\geq 1$, there is a deterministic $O(q)$-round algorithm to compute a weighted average $1/C$-relative defective $C$-coloring of $G$. The algorithm requires messages of size $O(\log C)$.
\end{lemma}
\begin{proof}
  The algorithm consists of $q$ phases, where in phase $i\in [q]$, the nodes with color $i$ (in the initial proper $q$-coloring) choose their color of the defective coloring. For the sake of the analysis, assume that all edges in the graph are oriented from the node with the larger initial color to the node with the smaller initial color. Note that because the initial $q$-coloring is a proper coloring, the orientation of every edge is well-defined. Consider some node $v\in V$ that chooses its color in $[C]$ in phase $i$. When $v$ chooses its color, all outgoing neighbors of $v$ have already chosen their color in previous phases (because they have a smaller initial color). Node $v$ chooses a color in $[C]$ such that the total weight of its monochromatic outgoing edges is minimized. Because $v$ can choose among $C$ different colors, there exists a color such that the total weight of the monochromatic outgoing edges is at most a $1/C$-fraction of the total weight of all outgoing edges. This directly implies that the computed coloring is weighted average $1/C$-relative defective. Clearly, one phase of the algorithm can be implemented in a single communication round and in order to choose its color, a node only needs to know the colors (of the defective coloring) of the already processed neighbors. Thus, the algorithm only requires messages of $O(\log C)$ bits.
\end{proof}

\subsection{Time-Efficient Weighted Average Defective Coloring Algorithms}

As a next step, we give a recursive algorithm that achieves a similar trade-off between relative average defect and number of colors more efficiently.

\begin{lemma}\label{lemma:simplerecursivedefective}
  Let $G=(V,E,w)$ be a weighted graph with non-negative edge weights and assume that we are given a proper vertex coloring with colors in $[q]$ of $G$. Then, for every $\eps\geq 1$, there is a deterministic $O(1/\eps + \log(1/\eps)\cdot\log^*q)$-round algorithm to compute a weighted  $\eps$-relative average defective $\lceil 2/\eps\rceil$-coloring of $G$. The algorithm requires messages of size $O(\log q)$.
\end{lemma}
\begin{proof}
  First note that we can assume that $q>\lceil 2/\eps\rceil$ as otherwise, we can just output the initial proper $q$-coloring as the defective coloring. For $\eps>0$, let $T(\eps)$ denote the time that is required to compute a weighted $\eps$-relative average defective $\lceil 2/\eps\rceil$-coloring of a given properly $q$-colored weighted graph. To prove the claimed time bound of the lemma, we thus need to show that $T(\eps)= O(1/\eps + \log(1/\eps)\cdot\log^*q)$. We argue that $T(\eps)$ can be phrased recursively as follows. There exists a constant $c>0$ such that
  \begin{equation}\label{eq:def_recursion}
    T(\eps) \leq \begin{cases}
      T(2\eps) + c/\eps + O(\log^* q) & \text{ if } \eps<1/2\\
      O(\log^* q) & \text{ if } \eps \geq 1/2
    \end{cases}
  \end{equation}
  Let us first consider the case $\eps\geq 1/2$. In this case, we can first compute a weighted $(\eps/2)$-relative defective $O(1)$-coloring in time $O(\log^* q)$ and with messages of size $O(\log q)$ bits by using \Cref{lemma:defcoloring}. Let $H$ be the subgraph of $G$ that only consists of the bichromatic edges of $G$ w.r.t.\ this $O(1)$-coloring. By using \Cref{lemma:simpleavgdefective} (for $C=\lceil 2/\eps\rceil$), we can then compute an $(\eps/2)$-relative average defective $\lceil 2/\eps\rceil$-coloring of $H$ in time $O(1/\eps)=O(1)$ with messages of size $O(\log(1/\eps))=O(1)$ bits. This coloring has a weighted average relative defect of at most $\eps$ and we have thus shown that $T(\eps)=O(\log^* q)$ for $\eps\geq 1/2$. As discussed, the algorithm only requires messages of $O(\log q)$ bits.
  
  For $\eps<1/2$, we first use \Cref{lemma:defcoloring} to compute a weighted $(1/4)$-relative defective coloring of $G$ with $s=O(1)$ colors. This can be done in time $O(\log^* q)$ with messages of size $O(\log q)$ bits. Let $G_1,\dots,G_s$ be the subgraphs induced by $s$ color classes of this coloring. Note that by definition, the total weight of all the edges in graph $G_1,\dots,G_s$ is at most $W(G)/4$, where $W(G)$ is the total weight of all edges in $G$. For each graph $G_i$, we now recursively compute a weighted $2\eps$-relative average defective $\lceil 1/\eps\rceil$-coloring in time $T(2\eps)$. When combining the $s$-coloring of $G$ with these $\lceil1/\eps\rceil$-colorings of each graph $G_i$, we obtain a $s\cdot\lceil 1/\eps\rceil$-coloring of $G$. Let $F\subseteq E$ be the set of monochromatic edges of this coloring. Note an edge $e$ is in $F$ if $e$ is an edge of one of the graphs $G_i$ for $i\in[s]$ and if $e$ is monochromatic in the weighted $2\eps$-relative average defective $\lceil 1/\eps\rceil$-coloring of $G_i$. Because because the total weight of the graphs $G_1,\dots,G_s$ is at most $W(G)/4$, the total weight of all edges in $F$ is at most $2\eps\cdot W(G)/4 = \eps/2\cdot W(G)$. Note also that by definition of the edge set $F$, the computed $s\cdot\lceil 1/\eps\rceil$-coloring of $G$ is a proper coloring of the graph $H:=(V,E\setminus F)$. We can therefore use \Cref{lemma:simpleavgdefective} to compute a weighted $(\eps/2)$-relative average defective $\lceil 2/\eps\rceil$-coloring of $H$ in time $O(s/\eps)$ and with messages of size $O(\log(1/\eps)) = O(\log q)$ bits. The total weight of all monochromatic edges of this coloring in $H$ is at most a $(\eps/2)$-fraction of the total weight of all edges in $H$ and thus also at most $\eps/2\cdot W(G)$. Together with the edges in $F$, which also could be monochromatic, we therefore get an $\lceil 2/\eps\rceil$-coloring, where the total weight of all monochromatic edges is at most $\eps\cdot W(G)$ as required. This therefore proves the \Cref{eq:def_recursion}.
  
  The claim of the lemma now follows directly from \Cref{eq:def_recursion} and from the fact that, as discussed above, throughout the recursive algorithm, we only used messages of $O(\log q)$ bits.
\end{proof}

\paragraph{Comment.} The argument in the above lemma can be seen as an adaptation of the recursive approach to solve $(\Delta+1)$-coloring in $O(\Delta+\log^* n)$ time in \cite{Kuhn2009WeakColoring}. We remark that alternatively, an essentially equivalent result to \Cref{lemma:recursiveavgdefective} could also be proven by adapting the arbdefective coloring algorithm of Barenboim, Elkin, and Goldenberg~\cite{BEG18} to the weighted setting. 
\medskip

\noindent We now have the tools  to prove \Cref{lemma:recursiveavgdefective}. For completeness, we first restate the lemma.

\LemmaRecAvgDefective*

\begin{proof}
  First note that w.l.o.g., we can assume that $q> C/\delta$, as otherwise, the claim of the lemma directly follows from \Cref{lemma:simpleavgdefective}.
  In the following, we use $W(G):=\sum_{e\in E} w(e)$ to denote the total weight of all edges of a weighted graph $G$. Recall that to prove the lemma, we need to give a distributed algorithm to compute a $C$-coloring of $G$ such that the total weight of all the monochromatic edges is at most $\frac{1+\delta}{C}\cdot W(G)$. 
  
  As a first step, we use \Cref{lemma:defcoloring} to compute an $O((C/\delta)^2)$-coloring of the nodes of $G$ such that the total weight of all monochromatic edges is at most $\frac{\delta}{2C}\cdot W(G)$. By \Cref{lemma:defcoloring}, such a coloring can be computed in $O(\log^*q)$ time and by using messages of $O(\log q)$ bits. Let $F\subseteq E$ be the set of monochromatic edges of this coloring and let $H:=(V,E\setminus F)$ be the subgraph of $G$ consisting of the edges that are not monochromatic. Note that because $H$ only has bichromatic edges, the defective $O((C/\delta)^2)$-coloring that we have computed for $G$ is a proper coloring of $H$. For convenience, we define $\delta':=\delta/2$. To prove the lemma, it now suffices to compute a weighted $(1+\delta')/C$-relative average defective $C$-coloring of $H$. Even if all the edges in $F$ end up being monochromatic w.r.t.\ to this coloring, the total weight of all monochromatic edges is still at most $\frac{(1+\delta')}{C}\cdot W(H) + w(F)$. From $W(H)\leq W(G)$, $\delta'=\delta/2$ and $w(F)\leq \frac{\delta}{2C}\cdot W(G)$, this immediately implies that the number of monochromatic edges is at most $\frac{1+\delta}{C}\cdot W(G)$.
  
  It therefore remains that we can compute a weighted $(1+\delta')/C$-relative average defective coloring of a graph $H$ with an initial proper $O((C/\delta)^2)$-coloring in the required round complexity and with the required message size. The algorithm to achieve this consists of two steps. First, we apply \Cref{lemma:simplerecursivedefective} to compute a weighted $\delta'/C$-relative defective coloring of $H$ with $\gamma=O(C/\delta')$ colors in time $O(C/\delta + \log(C/\delta)\cdot \log^*(C/\delta))=O(C/\delta)$ with messages of $O(\log(C/\delta))=O(\log q)$ bits. The total weight of all the monochromatic edges of this coloring is at most $\delta'/C\cdot W(H)$. Let $H'$ be the subgraph of $H$ that only consists of the bichromatic edges of this $\gamma$-coloring of $H$. We next apply \Cref{lemma:simpleavgdefective} to compute a weighted $1/C$-relative average defective $C$-coloring of $H'$ in time $O(\gamma)=O(C/\delta)$ and with messages of size $O(\log(C/\delta))=O(\log q)$ bits. When using this coloring on $H$, the total weight of all monochromatic edges is at most $W(H')/C + \delta'\cdot W(H)/C \leq (1+\delta')/C\cdot W(H)$ as required.
\end{proof}

\bibliographystyle{alpha}
\bibliography{ref}

\end{document}